\documentclass[sigconf,nonacm]{acmart}

\usepackage{soul}
\usepackage{url}
\usepackage[utf8]{inputenc}
\usepackage{amsmath}
\usepackage{amsthm}
\usepackage{mathtools}
\usepackage{algorithm}
\usepackage{algorithmic}
\usepackage{tabularx}
\usepackage{colortbl}
\usepackage{tikz}
\usetikzlibrary{matrix,fit,decorations.pathreplacing}
\pgfdeclarelayer{background}
\pgfdeclarelayer{foreground}
\pgfsetlayers{background,main,foreground}
\usepackage{subcaption}
\usepackage{cleveref}
\usepackage{xargs}
\usepackage{makecell}

\usepackage[disable]{todonotes}
\usepackage{tikz}
\usepackage{pgfplots}
\usetikzlibrary{shapes.misc,shapes.geometric}

\newtheorem{theorem}{Theorem}

\newtheorem{proposition}[theorem]{Proposition}
\newtheorem{lemma}[theorem]{Lemma}

\newtheorem{problem}{Problem}
\newtheorem*{claim}{Claim}{\upshape\itshape}{\upshape\rmfamily}
\newenvironment{claimproof}
{\noindent \textit{Proof of the claim.} }
{$\diamond$}

\Crefname{theorem}{Thm}{Thms}
\Crefname{proposition}{Prop}{Props}

\newcommand{\abs}[1]{{\vert #1 \vert}}
\newcommand{\hg}{\ensuremath{\mathcal{H}}}
\newcommand{\he}{\ensuremath{\mathcal{E}}}

\newcommand{\Oh}{\mathcal{O}}

\DeclareMathOperator*{\argmin}{arg\,min}

\newcommandx{\problemdef}[6][3=Input,5=Question]{
	\begingroup\par\noindent\nopagebreak[4]1
	\begin{problem}\label{prob:#2}\vspace{-0.5em}\hspace{-0.5em}\colorbox{white}{\begin{tabular}{@{}l@{}}\textsc{#1}\end{tabular}}\nopagebreak[4]\end{problem}\nopagebreak[4]
	\par\noindent\hangindent=\parindent\textbf{#3}:  #4\nopagebreak[4]
	\par\noindent\hangindent=\parindent\textbf{#5}:  #6
	\par\medskip
	\endgroup
}

\newcommand{\mbsp}{MBSP}

\newcommand{\NN}{\mathbb{N}} 
 
\newcommand{\Rp}{\mathbb{R}_{\geq 0}}

\newcommand{\headerL}[2]{\node[rectangle, rounded corners, draw=black, thick, inner sep=4pt, fill=black!10,xshift=12pt,yshift=-1pt,anchor=south west, minimum height=4.2ex] at (#1.north west) {#2};}

\newcommand{\paramBox}[4]{
		\node[inner sep = 2pt,anchor = south west] (A#1) at (#2) {
                  \begin{tabularx}{\TableWidth}{X!{\vrule width 1pt}X!{\vrule width 1pt}X!{\vrule width 1pt}X}
                    
			#4
		\end{tabularx}
		};
		\node[inner sep=-0.6pt,draw=white,ultra thick,rounded corners,fit=(A#1)] () {};
		\node[inner sep=-1.4pt,draw=white,ultra thick,rounded corners,fit=(A#1)] () {};
		\node[inner sep=-2pt,draw=black,thick,rounded corners,fit=(A#1)] (#1) {};
		\begin{pgfonlayer}{foreground}
			\headerL{#1}{#3}
		\end{pgfonlayer}
}

\title{The Influence of Agent Models on the Complexity of Bus Routing}

\author{Eva Deltl}
\orcid{https://orcid.org/0000-0003-0431-2733}
\affiliation{
	\institution{TU Clausthal}
	\city{Clausthal-Zellerfeld}
	\country{Germany}
}
\email{eva.deltl@tu-clausthal.de}

\author{Christian Komusiewicz}
\orcid{https://orcid.org/0000-0003-0829-7032}
\affiliation{
	\institution{Friedrich Schiller University}
	\city{Jena}
	\country{Germany}
}
\email{c.komusiewicz@uni-jena.de}

\author{Jurek Rostalsky}
\orcid{https://orcid.org/0009-0009-8818-4337}
\affiliation{
	\institution{Friedrich Schiller University}
	\city{Jena}
	\country{Germany}
}
\email{jurek.rostalsky@uni-jena.de}

\author{Johannes Schröder}
\orcid{https://orcid.org/0009-0003-8135-5785}
\affiliation{
	\institution{TU Berlin}
	\city{Berlin}
	\country{Germany}
}
\email{j.schroeder.1@tu-berlin.de}

\author{Luca Pascal Staus}
\orcid{https://orcid.org/0009-0004-3020-1011}
\affiliation{
	\institution{Friedrich Schiller University}
	\city{Jena}
	\country{Germany}
}
\email{luca.staus@uni-jena.de}

\begin{document}
	
\begin{abstract}
	In bus routing, the task is to plan a bus route in a network with several agents, each of whom wants to travel from a starting point to a destination. A bus route should account for several factors, including agents' cost for reaching the bus stops, their travel time, or the energy consumption of the buses.
	We study the complexity of several variants of this problem, focusing on how the objective function and the models for agents' walking costs influence the problem complexity. After observing that even the simplest agent cost model leads to hardness on general networks, we consider networks with tree structure.
	Our main findings are as follows. First, allowing agent-specific cost models leads to hardness even on extremely limited trees such as stars. Second, consistent agent models (where agents differ only in their starting points and destinations) make the problem easier in some cases. Finally, allowing agents to choose between using the bus and walking directly can make the problem considerably harder. Most of our hardness results show not only classical NP-hardness but also parameterized intractability for the natural parameter~$k$, the number of bus stops.
\end{abstract}

\maketitle

\section{Introduction}
Consider a city planner who wants to establish a bus line to serve transportation demands in a spatial network. In this planning task, numerous objectives and constraints need to be considered to reach a good solution. For example, the number of stops~$k$ on the bus line could be limited. Moreover, one would like to find a fast or energy-efficient route connecting these stops, which can be modelled by assigning edge weights~$w_b$ in the city's road network. Finally, to model the demands of the population, one would need to know the sources and destinations of the potential users of the bus line, called agents in the following. Ideally, we would also know how difficult or costly it is for each agent~$a$ to reach a given bus stop. This can be captured by an individual weight function~$w_a$ for each agent~$a$. These costs may vary drastically as some users may walk, with possibly drastically different walking speeds, while others use bicycles or scooters. With this information and an objective function~$f$ for weighing the different costs at hand, we face the following problem.

\smallskip
\noindent
  \textsc{Multi Bus Stop Problem (MBSP)}\\
  \textbf{Input:} {An undirected graph $G=(V,E)$, a set of agents $A$ with weight function $w_a:E\to\Rp$ for each $a\in A$, a bus weight function $w_b:E\to\Rp$,
    an integer $k\in\NN$, and an objective function $f$.}\\
  \textbf{Task:} {Compute a bus route $B$ with $k$ stops minimizing~$f(B)$.}\\[-.5em]

The first variant of this problem was stated by Reza et al.~\shortcite{originalBus} for the case where the route has only two bus stops and the task is to minimize the cost of the bus route plus the costs of the agents for the commute from source to the start of the bus route and from the end of the bus route to their destination. Since there are only two bus stops which are frequented by all agents, this cost function, later referred to as the \emph{energy} cost function, can fully incorporate the agent's travel time into the bus weight function~$w_b$. 

Recently, Proissl and Koch~\shortcite{DBLP:conf/gis/ProisslK24} extended Reza et al.'s model by allowing bus routes with up to~$k$ stops. In this setting, agents typically use the bus for only part of their journey, so their travel time cannot be incorporated into~$w_b$. Proissl and Koch therefore proposed the \emph{time} objective, which minimizes the total agent travel cost, including access to the boarding stop, travel along the bus route, and access from the exit stop to the destination.
As we will see, allowing multiple bus stops makes MBSP computationally hard in most settings. To address this, Proissl and Koch fix, for each agent~$a$, the indices of the stops where~$a$ boards and leaves the bus. Under this assumption, optimal bus stop placements can be computed in polynomial time~\cite{originalBus,DBLP:conf/gis/ProisslK24}.

Of course, it seems premature to fix the first and last bus stop for each agent, when it is not clear at all which route the bus will take. Motivated by this observation, we consider models where the agents enjoy greater freedom of choice. More precisely, in our setting the agents choose optimal stops for entering and leaving the bus after the bus stops have been placed. Since the resulting bus route may not be convenient for every agent, we also consider a second variant of each objective function in which agents are not required to take the bus and may instead walk directly to their destinations. Altogether, this yields four objective functions~$f_\text{energy}$, $f_\text{time}$, $g_\text{energy}$, and~$g_\text{time}$, where the latter two functions model the case in which agents can walk directly.

With these extensions, MBSP becomes closely related to problems such as \textsc{$k$-Median}~\cite{HK79} and \textsc{Facility Location}~\cite{CNW83,FELLOWS20111118}, which roughly correspond to the case where the bus weight~$w_b$ is zero for all edges. Consequently, MBSP is NP-hard and also intractable in the parameterized sense with respect to the number~$k$ of bus stops. Thus, algorithms with running time~$f(k)\cdot n^{\Oh(1)}$ for input size~$n$ are unlikely. 
We therefore study how the complexity of \mbsp{} depends on two factors: the structure of the underlying spatial network and the mobility model of the agents. In our setting, the input graph represents a transportation network, and agents represent origin-destination demands. We analyze general graphs, trees, paths, and stars, and identify which combinations of network structure and agent model admit efficient exact algorithms.
To illustrate that the objectives lead to different spatial stop patterns in practice, we optimize the energy objectives on the NYC-M15 bus corridor, using existing stops as candidates.\footnote{To support reproducibility, code, preprocessing scripts, and the processed NYC-M15 instance are available at \url{https://github.com/evamichelle30/ComplexityBusRouting}}

\subsection{Our Results}
We use paths, stars, and trees as abstractions of common spatial network structures: paths model corridor-like routes, stars model access through a central hub, and trees model branching networks. In a nutshell, we show that natural variants of \mbsp{} are already hard on paths and stars and even for restricted agent mobility models, while combining network restrictions with simple mobility models can lead to tractability; see~\Cref{fig:overview}.

\begin{figure*}[t]

	\definecolor{ColorOpen}{RGB}{255, 255, 255}
	\definecolor{ColorW2h}{RGB}{255, 100, 100} 
	\definecolor{ColorW1h}{RGB}{255, 200, 100} 
	\definecolor{ColorNPh}{RGB}{255, 255, 100} 
	\definecolor{ColorP}{RGB}{150, 200, 50}    
	\definecolor{ColorUnclear}{RGB}{180,180,180}

	\newcolumntype{C}[1]{>{\centering\let\newline\\\arraybackslash\hspace{0pt}}p{#1}}
	\newcolumntype{D}{>{\centering\let\newline\\\arraybackslash\hspace{0pt}}X}

	\newcolumntype{b}{>{\hsize=1.3\hsize}D}
	\newcolumntype{m}{>{\hsize=1\hsize}D}
	\newcolumntype{s}{>{\hsize=.7\hsize}D}

	\newcommand{\columnSep}{0.26cm}
	\newcommand{\rowDist}{1.6cm}

	\newcommand{\TableWidth}{4.4cm}
	\centering \small
	\begin{tikzpicture}[xscale=0.95,yscale=0.95]

		\begin{scope}[]
			\paramBox{gen}{0,0}{General}{
				\cellcolor{ColorW2h} &
				\cellcolor{ColorW2h} &
				\cellcolor{ColorW2h} &
				\cellcolor{ColorW2h}\vspace{-0.65em}\hspace*{-0.3em}\Cref{thm:W2_on_unweighted_splitgraphs} \\
			}

			\paramBox{trees}{0,-1.15}{Trees}{
				\cellcolor{ColorW2h} &
				\cellcolor{ColorP}\vspace{-0.65em}\hspace*{-0.3em}\Cref{thm:f_energy_poly_on_consistent_trees} &
				\cellcolor{ColorP} &
				\cellcolor{ColorP} \\
			}

			\paramBox{paths}{-2,-2.85}{Paths}{
				\cellcolor{ColorP}\vspace{-0.65em}\hspace*{-0.3em}\Cref{thm:energy_poly_on_arbitrary_paths}  &
				\cellcolor{ColorP}&
				\cellcolor{ColorP} &
				\cellcolor{ColorP} \\
			}

			\paramBox{stars}{2,-2.3}{Stars}{
				\cellcolor{ColorW2h}\vspace{-0.65em}\hspace*{-0.3em}\Cref{thm:f_W2_on_arbitrarily_weighted_stars} &
				\cellcolor{ColorP} &
				\cellcolor{ColorP}&
				\cellcolor{ColorP}\\
			}

			\draw[thick] (trees) -- (gen);
			\draw[thick] (trees) -- (paths);
			\draw[thick] (trees) -- (stars);

			\node[scale=1.25, right] at (-2, 0) {$f_\text{energy}$};
		\end{scope}

		\begin{scope}[xshift=9.5cm]
			\paramBox{gen}{0,0}{General}{
				\cellcolor{ColorW2h} &
				\cellcolor{ColorW2h} &
				\cellcolor{ColorW2h} &
				\cellcolor{ColorW2h}\vspace{-0.65em}\hspace*{-0.3em}\Cref{thm:W2_on_unweighted_splitgraphs} \\
			}

			\paramBox{trees}{0,-1.15}{Trees}{
				\cellcolor{ColorW1h} &
				\cellcolor{ColorW1h} &
				\cellcolor{ColorW1h} &
				\cellcolor{ColorW1h} \\
			}

			\paramBox{paths}{-2,-2.85}{Paths}{
				\cellcolor{ColorP}\vspace{-0.65em}\hspace*{-0.3em}\Cref{thm:energy_poly_on_arbitrary_paths}  &
				\cellcolor{ColorP}&
				\cellcolor{ColorP} &
				\cellcolor{ColorP} \\
			}

			\paramBox{stars}{2,-2.3}{Stars}{
				\cellcolor{ColorW1h} &
				\cellcolor{ColorW1h} &
				\cellcolor{ColorW1h}&
				\cellcolor{ColorW1h}\vspace{-0.65em}\hspace*{-0.3em}\Cref{thm:g_energy_hard_on_unweighted_stars}\\
			}

			\draw[thick] (trees) -- (gen);
			\draw[thick] (trees) -- (paths);
			\draw[thick] (trees) -- (stars);

			\node[scale=1.25, right] at (-2, 0) {$g_\text{energy}$};
		\end{scope}

		\begin{scope}[yshift=-4.4cm]
			\paramBox{gen}{0,0}{General}{
				\cellcolor{ColorW2h} &
				\cellcolor{ColorW2h} &
				\cellcolor{ColorW2h} &
				\cellcolor{ColorW2h}\vspace{-0.65em}\hspace*{-0.3em}\Cref{thm:W2_on_unweighted_splitgraphs} \\
			}

			\paramBox{trees}{0,-1.15}{Trees}{
				\cellcolor{ColorW1h} &
				\cellcolor{ColorW1h} &
				\cellcolor{ColorW1h} &
				\cellcolor{ColorNPh}\vspace{-0.65em}\hspace*{-0.5em}\Cref{thm:f_time_hard_on_unweighted_trees} \\
			}

			\paramBox{paths}{-2,-2.85}{Paths}{
				\cellcolor{ColorW1h}\vspace{-0.65em}\hspace*{-0.4em}\Cref{thm:time_hard_on_arbitrary_paths} &
				\cellcolor{ColorUnclear} &
				\cellcolor{ColorUnclear} &
				\cellcolor{ColorP}\vspace{-0.65em}\hspace*{-0.5em}\Cref{thm:f_time_poly_on_unweighted_paths} \\
			}

			\paramBox{stars}{2,-2.3}{Stars}{
				\cellcolor{ColorW1h} &
				\cellcolor{ColorW1h} &
				\cellcolor{ColorW1h}\vspace{-0.65em}\hspace*{-0.3em}\Cref{thm:time_hard_on_unit_stars}&
				\cellcolor{ColorP}\vspace{-0.65em}\hspace*{-0.3em}\Cref{thm:f_time_poly_on_unweighted_stars} \\
			}

			\draw[thick] (trees) -- (gen);
			\draw[thick] (trees) -- (paths);
			\draw[thick] (trees) -- (stars);

			\node[scale=1.25, right] at (-2, 0) {$f_\text{time}$};
		\end{scope}

		\begin{scope}[xshift=9.5cm, yshift=-4.4cm]
			\paramBox{gen}{0,0}{General}{
				\cellcolor{ColorW1h} &
				\cellcolor{ColorW1h} &
				\cellcolor{ColorW1h} &
				\cellcolor{ColorP}\vspace{-0.65em}\hspace*{-0.3em}Prop~\ref{prop:g_time_poly_on_unweighted_graphs} \\
			}

			\paramBox{trees}{0,-1.15}{Trees}{
				\cellcolor{ColorW1h} &
				\cellcolor{ColorW1h} &
				\cellcolor{ColorW1h} &
				\cellcolor{ColorP} \\
			}

			\paramBox{paths}{-2,-2.85}{Paths}{
				\cellcolor{ColorW1h}\vspace{-0.65em}\hspace*{-0.4em}\Cref{thm:time_hard_on_arbitrary_paths} &
				\cellcolor{ColorUnclear} &
				\cellcolor{ColorUnclear} &
				\cellcolor{ColorP} \\
			}

			\paramBox{stars}{2,-2.3}{Stars}{
				\cellcolor{ColorW1h} &
				\cellcolor{ColorW1h} &
				\cellcolor{ColorW1h}\vspace{-0.65em}\hspace*{-0.3em}\Cref{thm:time_hard_on_unit_stars}&
				\cellcolor{ColorP}\\
			}

			\draw[thick] (trees) -- (gen);
			\draw[thick] (trees) -- (paths);
			\draw[thick] (trees) -- (stars);

			\node[scale=1.25, right] at (-2, 0) {$g_\text{time}$};
		\end{scope}

		\begin{scope}[xshift=-1cm, yshift=-4.3cm]
			\renewcommand{\TableWidth}{7.5cm}
			\paramBox{legend}{1.1,-3*\rowDist}{Graph restriction}{
				\centering arbitrary  &
				\centering consistent &
				\centering unit &
				\centering unweighted 
			};
		\end{scope}

		\begin{scope}[xshift=-4cm, yshift=-4.3cm]
			\node[draw=black,very thick,inner sep = 0pt,anchor=south east] at (\textwidth,-3*\rowDist + 2pt) {
				\begin{tabularx}{5.2cm}{D}
					\cellcolor{ColorW2h}NP-h and W[2]-h wrt.~$k$\\ \hline 
					\cellcolor{ColorW1h} NP-h and W[1]-h wrt.~$k$ \\ \hline
					\cellcolor{ColorNPh} NP-h \\ \hline
					\cellcolor{ColorP} Poly-time \\ \hline
				\end{tabularx}
			};
		\end{scope}
		
	\end{tikzpicture}

	\caption{
		Overview of our results. 
              }
	\label{fig:overview}
\end{figure*}
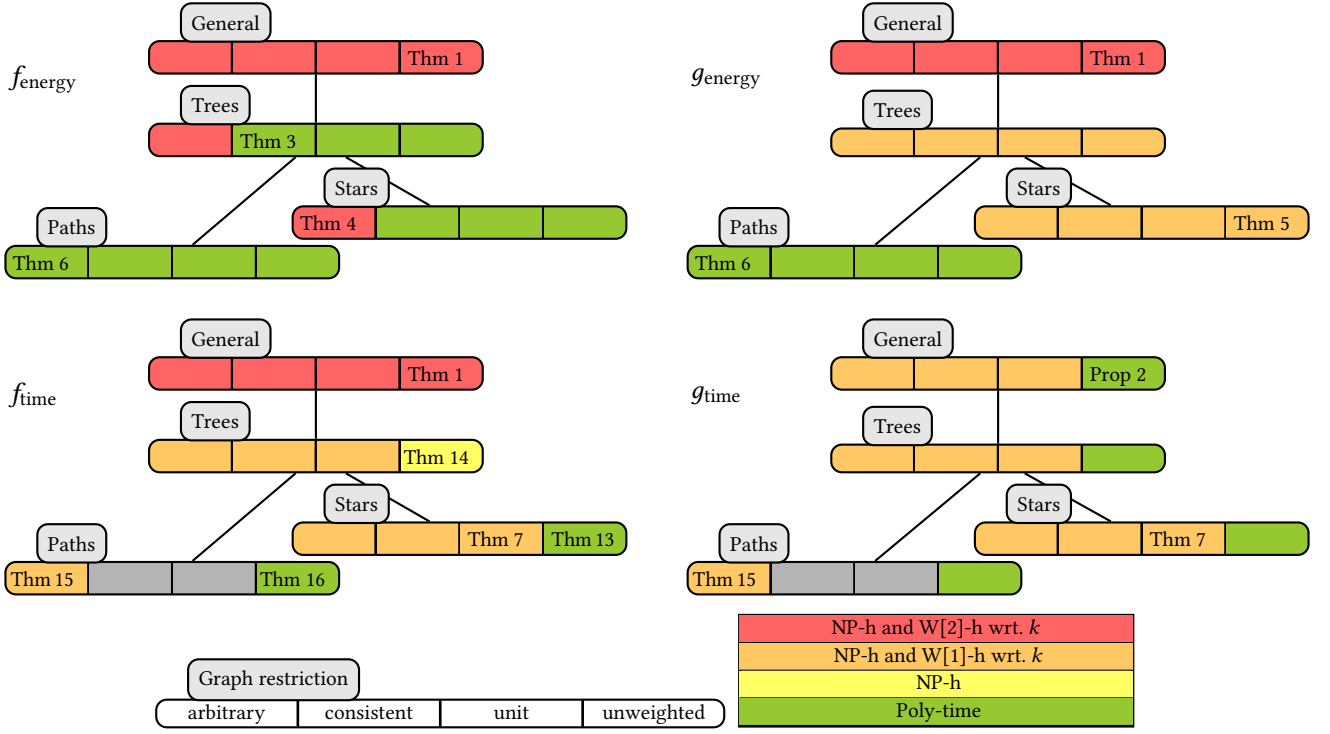
In terms of agent mobility models, represented formally by agent-specific weight functions, we consider four cases.
\begin{itemize}
	\item \textbf{Arbitrary:} each agent has its own weight function $w_a$ and $w_b$ is arbitrary.
	\item \textbf{Consistent:} all agents share the same weight function.
	\item \textbf{Unit:} all agent edge weights are equal to a constant $c\in\NN$ (equivalently, $w_a\equiv 1$ up to scaling).
	\item \textbf{Unweighted:} all agent and bus edge weights are $1$.
\end{itemize}

For~$g_\text{time}$, the unweighted case is uninteresting since it is optimal for all agents to walk directly to their destination; we exclude this case from the following discussion.
For all other objective functions, we show that the unweighted case is hard on general graphs. Note that this hardness does \emph{not} follow from the above-mentioned connection to \textsc{$k$-Center}, since that connection applies to the case where all bus weights are zero.
This hardness motivates us to consider restricted input graphs and assess, for each, how realistic the agent models can be before facing intractability.

We show for example that for~$f_\text{energy}$, MBSP is polynomial-time solvable on trees for consistent agent weights and hard for arbitrary agent weights. In contrast, for~$f_\text{time}$ all weight function models yield hard problems on trees. This indicates that minimizing~$f_\text{energy}$ is easier than minimizing~$f_\text{time}$. The same can be observed for the more restricted cases where the input is a path or a star.

Allowing agents to choose whether to take the bus also makes the problem harder: minimizing~$g_\text{energy}$ is already hard on stars in the unweighted case. On paths, however, ~$f_\text{energy}$ and~$g_\text{energy}$ are tractable even for arbitrary weight models. This case is particularly interesting because it corresponds to optimizing bus stop placement along an existing bus route.
To illustrate this application, we consider the NYC-M15 corridor, showing that the two objectives can lead to different stop sets.

\subsection{Further Related Work} 
The~\textsc{$k$-Median} problem in graphs asks to select~$k$ vertices in a graph (corresponding to the bus stops) such that the sum of the minimum distances of the graph vertices to the selected vertices is minimized. In network analysis, this problem is also known as~\textsc{Group Closeness Centrality}~\cite{EP99}. It is polynomial-time solvable on trees~\cite{Tamir96} but W[2]-hard with respect to~$k$ on bounded-degree graphs~\cite{SKMS23}. 
The most crucial differences to our setting are that there are no costs for connecting the selected vertices, in other words, no bus \emph{route}; \emph{every} graph vertex is an agent that needs to reach one of the selected vertices; and all agents have the same weight function.

The \textsc{Facility Location} problem, in particular \textsc{Connected Facility Location}~\cite{Ljubic20} with fixed number~$k$ of facilities is also closely related. Here, the selected facilities correspond to our bus stops but the connection between them is established via a Steiner Tree, not a route. Even more crucially, agents only seek to reach a facility; travel time is not modeled.
A more agent-centric study on \textsc{Facility Location} was undertaken by Aziz et al.~\shortcite{ACL+20}; unlike our work they consider the problem with capacity constraints on the facilities and problem versions where the agents do not reveal their locations. 

Another line of related work concerns set covering routing problems~\cite{MORADI2024110730},
where one must choose which vertices to visit and how to connect them. Among the
most studied variants are the \textsc{Travelling Purchaser Problem} (TPP) and
the \textsc{Ring Star Problem} (RS).
TPP asks for a minimum-cost tour visiting
markets where all required products can be bought~\cite{MANERBA2017,XiaoEtal2022}.
MBSP with $f_{\mathrm{energy}}$ is a special case: markets correspond to
potential bus stops, purchase costs encode agents' contributions, and the tour
encodes the bus route. Since the reduction preserves the underlying graph
structure, our hardness results for restricted graph classes, such as stars,
transfer directly to TPP.
In contrast, RS asks for a cycle through selected vertices, assigning all remaining vertices to them at minimum routing and assignment cost~\cite{ringstar1}. Unlike MBSP, it uses a cycle rather than a route and does not model individual source--destination pairs or travel-time preferences.

As we consider agents with start point and destination, the problem can also be considered as a pick up and delivery problem; for a survey see~\cite{KOC2020104987}. A further related problem is \textsc{School Bus Routing}~\cite{PK10} which often addresses the case where many routes must be planned simultaneously and considers a plethora of additional aspects such as distribution of the students to bus routes, adjustment of bus timeplans and school bells. These problems are considerably harder than our setting and solved with heuristics~\cite{AG20,SKS+13}. 

Finally, MBSP can be viewed as a stylized variant of \emph{transit network design}, which studies the design of public transport routes, stops, frequencies, and schedules~\citep{GuihaireHao08,Owais26,DerribleKennedy11}. 
More specifically, it is also related to work on bus-stop location and spacing~\citep{FurthRahbee00}. In contrast, we study the computational complexity of a single-route stop-selection problem with source--destination agents and agent-specific access costs.

\section{Preliminaries}
For~$n \in \NN$, we denote by~$[n]$ the set~$\{1,2,\ldots,n\}$.
\subsubsection*{Graphs and Shortest Walks}
We use standard notation from graph theory.
Let $G=(V,E)$ be an undirected graph. A \emph{walk} in $G$ is a sequence of vertices
$P=(v_1,\dots,v_\ell)$ such that $\{v_i,v_{i+1}\}\in E$ for all $i\in [l-1]$.
We say that $P$ \emph{starts} at $v_1$ and \emph{ends} at $v_\ell$.
A \emph{path} is a walk in which each vertex appears at most once.
Let $w:E\to\Rp$ be an edge-weight function.
The \emph{weight} of a walk $P=(v_1,\dots,v_\ell)$ is
\[
w(P) \ :=\ \sum_{i=1}^{\ell-1} w(\{v_i,v_{i+1}\}).
\]
For vertices $s,t\in V$, we denote by $\Pi(s,t,w)$ a minimum-weight walk from $s$ to $t$
under weights $w$. Note that $\Pi(s,t,w)$ can always be chosen to be a path.
We write~$\pi(s,t,w) \ :=\ w(\Pi(s,t,w))$ for the corresponding minimum weight.

\subsubsection*{Agents and Bus Routes}

An \emph{agent} is a tuple $a=(s,t,w)$ with $s,t\in V$ and $w:E\to\Rp$, where $s$ is
the agent's start vertex, $t$ is the agent's destination, and $w$ is the agent's edge-weight
function. We use $s_a$, $t_a$, and $w_a$ to refer to the attributes of $a$, respectively.
A \emph{bus route} with $k$~stops is a sequence of vertices~$B=(v_1,\dots,v_k)$.
Consecutive stops in $B$ need not be adjacent in~$G$ and vertices may repeat.
We write~$B[i]$ for the $i$th stop of~$B$.

Given indices $i,j\in[k]$ and bus weights $w_b$, we define $\Pi_B(i,j)\coloneqq \Pi(B[i],B[j],w_b)$ as a minimum-weight
walk
that starts at $B[i]$, ends at $B[j]$, and visits the stops of the route in order along the
subsequence from $B[i]$ to $B[j]$. More precisely, the walk must visit in order 
\[
B[i],B[i+1],\dots,B[j] \quad \text{or} \quad B[i],B[i-1],\dots,B[j]
\]
depending on whether $i\le j$ or $i>j$.
We write~$\pi_B(i,j) \ :=\ w_b(\Pi_B(i,j))$ for the minimum weight of such a constrained walk.
For an agent~$a \in A$, the walking cost of boarding at~$B[i]$ and leaving at~$B[j]$ is
\[
\pi_a(B,i,j)
:=
\pi(s_a,B[i],w_a)
+
\pi(B[j],t_a,w_a).
\]

\subsubsection*{Objective Functions}
We study the following two objectives.

The \emph{energy objective}, considered in \cite{originalBus} and \cite{DBLP:conf/gis/ProisslK24}, consists of the bus travel cost along the full route plus the sum of walking costs to and from the bus of all agents:
\[
f_{\text{energy}}(B)
\ :=\ 
\pi_B(1,k)
\ +\ 
\sum_{a\in A}\ 
\min_{i,j\in\{1,\dots,k\}}
\Bigl(
\pi_a(B,i,j)
\Bigr).
\]

The \emph{time objective} models that each agent additionally experiences the bus travel time between its boarding and exit stops:
\[
f_{\text{time}}(B)
\ :=\ 
\sum_{a\in A}\ 
\min_{i,j\in\{1,\dots,k\}}
\Bigl(
\pi_a(B,i,j)
+
\pi_B(i,j)
\Bigr).
\]
Note that \(f_\text{time}\) as defined above differs from all time objective functions of Proissl and Koch~\shortcite{DBLP:conf/gis/ProisslK24}. They model time as average travel time including pre-departure waiting, assuming that all agents start simultaneously towards a fixed start stop and that the bus departs only after the last agent arrives. In our setting, agents choose boarding and exit stops autonomously, so we do not model such agent coordination and aggregate per-agent travel times.

\paragraph{Variants allowing direct walking.}
We also consider variants in which an agent may be faster by walking directly to its destination:
\begin{itemize}
	\item $g_{\text{energy}}$ is the same as $f_{\text{energy}}$, except that each agent may alternatively prefer
	$\pi(s_a,t_a,w_a)$;
	\item $g_{\text{time}}$ is the same as $f_{\text{time}}$, except that each agent may alternatively prefer
	$\pi(s_a,t_a,w_a)$.
\end{itemize}

\subsubsection*{Parameterized Complexity}
All four objective functions can clearly be computed in polynomial time for a given bus route, which means that for each variant there is a brute-force algorithm with running time~$n^{k+\Oh(1)}$. For large~$n$, this running time becomes prohibitively large even for rather small values of~$k$. In contrast, FPT-algorithms for~$k$, that is, algorithms running in time~$f(k)\cdot n^{\Oh(1)}$ are more desirable in practice. To exclude not only polynomial running times in our lower bounds, but also FPT-running times, we use the framework of parameterized complexity~\cite{DF13,CFK+15} where each input instance consists of the classical instance and a problem-specific parameter~$k$. Basic classes of intractability are W[1]-hardness and W[2]-hardness. In other words, it is assumed that W[1]-hard and W[2]-hard problems do not admit FPT-algorithms.
To show W[1]-hardness or W[2]-hardness of a problem, one describes a parameterized reduction from a W[1]-hard or W[2]-hard problem, respectively. A \emph{parameterized reduction} is a reduction that runs in FPT-time and for which the parameter of the constructed instance depends only on the parameter of the input instance; for a formal definition, we refer the reader to the standard textbooks on parameterized complexity~\cite{DF13,CFK+15}. 

\section{Hardness on General Unweighted Graphs}
\label{sec:general_unweighted}

To motivate the subsequent studies where we look at the hardness of \mbsp{} on trees, stars, and paths, we first show that all objective functions except~$g_{\text{time}}$ are W[2]-hard for the parameter~$k$ on general graphs even in the unweighted case. For this we use reductions from \textsc{Hitting Set}.
A \textsc{Hitting Set} instance consists of a hypergraph~\(\hg = (V, \he)\) and a budget~\(k\). The goal is to decide if there is a \emph{hitting set}~\(S \subseteq V\) with~\(\abs{S} = k\) and~\(S \cap e \neq \emptyset\) for all~\(e \in \he\).
\textsc{Hitting Set} is W[2]-hard for the parameter~$k$~\cite{DF13}.

\begin{theorem}\label{thm:W2_on_unweighted_splitgraphs}
	\mbsp{} with objective function~\(f_{\text{energy}}\),~\(g_{\text{energy}}\), or \(f_{\text{time}}\) is W[2]-hard for~\(k\) on general graphs in the unweighted case.
\end{theorem}
\begin{proof}
  We first show the statement for~\(g_{\text{energy}}\).
  Given an instance \((\hg = (V, \he),k)\) We construct an \mbsp{} instance on a graph~\(G = (V^\prime, E)\) as follows:
For each vertex~\(v \in V\) there is a vertex~\(x_v \in V^\prime\), and we connect all those vertices to form a clique in~\(G\).
Furthermore, for each edge~\(e \in \he\) we add a vertex~\(s_{e}\) to~\(V^\prime\) and for all~\(v \in e\) we add edges~\(\{s_e, x_v\}\) to~\(E\).
Additionally, we introduce a destination vertex~\(t\), which is connected to \(x_v\) for each~\(v \in V\).
There is an agent~\(a_e\) for each edge~\(e \in \he\) who travels from~\(s_e\) to~\(t\). We consider the unweighted case, that is, we set all edge weights for the agents and the bus to~$1$. We set the number of bus stops to~$k'\coloneqq k + 1$.

	We next show the correctness of the reduction.

	\emph{Claim}: \(\hg\) has a hitting set of size~\(k\) if and only if the \mbsp{} instance has a solution~$B$ with~\(g_{\text{energy}} \leq \abs{\he} + k\).

	\((\Rightarrow)\) Let~\(S\) be a hitting set of size~\(k\) in~\(\hg\).
	Then, by placing a bus stop at~\(t\) and at every~\(x_v\) with~\(v \in S\), we obtain a bus route~$B$ with~\(k + 1\) bus stops.
	The objective function value of~$B$ is at most~\(\abs{\he} + k\): for every agent walking one edge suffices to reach the closest bus stop and drive to~$t$, and~\(k\) edges connect the~\(k + 1\) bus stops.

	\((\Leftarrow)\) Let~$B$ be an optimal solution of the MBSP instance.
	We show that $t$ is a bus stop and that for each edge~$e\in \he$ at least one of the vertices in~$N(s_e)\cup \{ s_e \}$ is a bus stop. Then, we can construct a hitting set of size at most~\(k\) by first adding~$v\in V$ to the hitting set if~$x_v$ is a bus stop and then adding one vertex from each edge~$e\in\he$ if~$s_e$ is a bus stop.
	
	Suppose some vertex is the location of two bus stops. Then either all agents are adjacent to a bus stop and \(t\) is a bus stop, or we can move a duplicate bus stop to a vertex adjacent to an agent who is not yet adjacent to a bus stop or to \(t\) without increasing the objective value since we save a cost of at least 1 for the agents while prolonging the bus tour by at most 1.
	We can repeat this process until there are no duplicate bus stop locations or all agents are adjacent to a bus stop and \(t\) is a bus stop.
	For the rest of the proof we assume that there are no two bus stops at the same vertex.

	If for some edge~$e\in \he$ there is a bus stop at~$s_e$ but not in~$N(s_e)$, then we can move the bus stop from~$s_e$ to an arbitrary vertex in~$N(s_e)$. This increases the cost of agent~$a_e$ by~$1$ but decreases the bus cost by at least~$1$.
  Connecting the~\(k + 1\) bus stops contributes at least~\(k\) to the objective function.
	Hence, the total contribution of the agents is at most~\(\abs{\he}\).
  If there is an agent~\(a_e\) who contributes more than~\(1\) to the objective function value, then there must be another agent~\(a_{e^\prime}\) who contributes~\(0\) to the objective function value.
	Note that this implies that there must be a bus stop in~\(s_{e^\prime}\) and therefore also in~\(N(s_{e^\prime})\).
	We can move this bus stop from~\(s_{e^\prime}\) to a neighbor of~\(s_e\) without increasing the objective function value because the contribution of at least agent~\(a_e\) is decreased by at least~$1$ and the contribution of agent~\(a_{e^\prime}\) is increased by~\(1\).
  If there are no more agents who contribute more than~\(1\) to the objective function value, then for every edge~$e\in \he$ there must be at least one bus stop in~$N(s_e)\cup \{ s_e \}$. To construct a hitting set of size at most~$k$ we need to look at two cases. In the first case~$t$ is a bus stop and we are done. In the second case~$t$ is not a bus stop. This implies that for all~$e\in \he$ the vertex~$s_e$ is a bus stop since the agent~$a_e$ contributes exactly~$1$ to the objective function by walking over an edge that leads to~$t$. The bus route must therefore contribute at most~$k$ to the objective function which is only possible if every agent $a_e$ has a bus stop in $N(s_e)$. This implies that there are at most $k$ edges in $\he$ which gives us a trivial solution.

Altogether, this shows the correctness of the reduction. Since the construction of the \mbsp{} instance forces all agents to use the bus in the optimal solution, the result carries directly over to \(f_\text{energy}\).

To show hardness for~$f_{\text{time}}$, we slightly modify the construction: First, each agent \(a_e\) now has its own destination vertex \(t_e\) which has the same neighborhood as \(s_e\). Second, we set the number~$k'$ of bus stops to~$k$.
	It remains to show the following.

	\emph{Claim}:  \(\hg\) has a hitting set of size~\(k\) if and only if the \mbsp{} instance has a solution~$B$ with~\(f_{\text{time}} \leq 2\abs{\he}\).

	\((\Rightarrow)\) Let~\(S\) be a hitting set of size~\(k\) in~\(\hg\).
	Then, placing a bus stop at~\(x_v \in V^\prime\) for all~\(v \in S\) is a solution to the \mbsp{} with~\(k\) bus stops on~\(G\). 
	By construction of the \mbsp{} instance on~\(G\) and the definition of~\(S\) each agent contributes~\(2\) to the objective function value, which sums up to~\(2\abs{\he}\).

	\((\Leftarrow)\) Let~$B$ be an optimal solution of the \mbsp{} instance. Because no agent can contribute less than~\(2\) to the objective function value, all agents must contribute exactly~\(2\).
	This implies that for each agent~\(a_e\) at least one of the vertices in~\(N(s_e) \cup \{s_e, t_e\}\) must be a bus stop. We can use this to easily construct a hitting set of size at most~$k$ by adding \(x_v\) to \(S\) if \(x_v\) is a bus stop and adding a vertex from \(N(s_e)\) to \(S\) whenever \(s_e\) or \(t_e\) is a bus stop.
\end{proof}

In contrast to Theorem~\ref{thm:W2_on_unweighted_splitgraphs}, \mbsp{} is trivial with the~$g_{\text{time}}$ objective function. This is because in the unweighted case an agent never benefits from using the bus as they can always just walk directly to their end point. Consequently, any bus route with~$k$ stops is an optimal solution, which gives the following.

\begin{proposition}\label{prop:g_time_poly_on_unweighted_graphs}
  \mbsp{} with objective function~\(g_\text{time}\) can be solved in polynomial time on general graphs in the unweighted case.
\end{proposition}

\section{Minimizing Energy on Trees}
\label{sec:energy}

\begin{figure*}[t]
	\centering
		\begin{subfigure}[t]{0.33\textwidth}
	\centering
	
	\begin{tikzpicture}[yscale=0.65]
							
		\node[fill,inner sep=2pt,minimum size=0pt] (R) at (0,0) [circle,draw] {};
		\node at (R) [above,yshift=1mm] {\large $r$};
		\node[fill,inner sep=2pt,minimum size=0pt] (A) at (-2,-1.5) [circle,draw] {};
		\node at (A) [left,xshift=-1mm] {\large $a$};
		\node[fill,inner sep=2pt,minimum size=0pt] (B) at (0,-1.5) [circle,draw] {};
		\node at (B) [left,xshift=-1mm] {\large $b$};
		\node[fill,inner sep=2pt,minimum size=0pt] (C) at (2,-1.5) [circle,draw] {};
		\node at (C) [left,xshift=-1.5mm] {\large $c$};
		\node[fill,inner sep=2pt,minimum size=0pt] (D) at (2,-3) [circle,draw] {};
		\node at (D) [left,xshift=-1.5mm] {\large $d$};

		\draw (R) -- (A);
		\draw (R) -- (B);
		\draw (R) -- (C);
		\draw (C) -- (D);
		
	\end{tikzpicture}
	\caption{After rooting the tree at $r$.}
	\label{fig:tree_dp_example_a}
	\end{subfigure}
		\begin{subfigure}[t]{0.33\textwidth}
	\centering
		
	\begin{tikzpicture}[yscale=0.65]
							
		\node[fill,inner sep=2pt,minimum size=0pt] (R) at (0,0) [circle,draw] {};
		\node at (R) [above,yshift=1mm] {\large $r$};
		\node[inner sep=2pt,minimum size=0pt] (CR) at (-2.25,-1.5) [circle,draw] {};
		\node at (CR) [left,xshift=-1mm] {\large $c_r$};
		\node[fill,inner sep=2pt,minimum size=0pt] (A) at (-0.75,-1.5) [circle,draw] {};
		\node at (A) [left,xshift=-1mm] {\large $a$};
		\node[fill,inner sep=2pt,minimum size=0pt] (B) at (0.75,-1.5) [circle,draw] {};
		\node at (B) [left,xshift=-1mm] {\large $b$};
		\node[fill,inner sep=2pt,minimum size=0pt] (C) at (2.25,-1.5) [circle,draw] {};
		\node at (C) [left,xshift=-1.5mm] {\large $c$};
		\node[inner sep=2pt,minimum size=0pt] (CC) at (1.75,-3) [circle,draw] {};
		\node at (CC) [left,xshift=-1mm] {\large $c_c$};
		\node[fill,inner sep=2pt,minimum size=0pt] (D) at (2.75,-3) [circle,draw] {};
		\node at (D) [left,xshift=-1mm] {\large $d$};

		\draw (R) -- (A);
		\draw (R) -- (B);
		\draw (R) -- (C);
		\draw[dashed] (R) -- (CR);
		\draw (C) -- (D);
		\draw[dashed] (C) -- (CC);
		
	\end{tikzpicture}
	\caption{After adding dummy children.}
	\label{fig:tree_dp_example_b}
	\end{subfigure}
		\begin{subfigure}[t]{0.33\textwidth}
	\centering
		
	\begin{tikzpicture}[yscale=0.65]
							
		\node[fill,inner sep=2pt,minimum size=0pt] (R) at (0,0) [circle,draw] {};
		\node at (R) [above,yshift=1mm] {\large $r$};
		\node[inner sep=2pt,minimum size=0pt] (CR) at (-0.75,-0.75) [circle,draw] {};
		\node at (CR) [left,xshift=-1mm] {\large $c_r$};
		\node[inner sep=2pt,minimum size=0pt] (X2) at (0.75,-0.75) [circle,draw] {};
		\node at (X2) [left,xshift=-1mm] {\large $x_2$};
		\node[fill,inner sep=2pt,minimum size=0pt] (A) at (0,-1.5) [circle,draw] {};
		\node at (A) [left,xshift=-1mm] {\large $a$};
		\node[inner sep=2pt,minimum size=0pt] (X1) at (1.5,-1.5) [circle,draw] {};
		\node at (X1) [left,xshift=-1mm] {\large $x_1$};
		\node[fill,inner sep=2pt,minimum size=0pt] (B) at (0.75,-2.25) [circle,draw] {};
		\node at (B) [left,xshift=-1mm] {\large $b$};
		\node[fill,inner sep=2pt,minimum size=0pt] (C) at (2.25,-2.25) [circle,draw] {};
		\node at (C) [left,xshift=-1mm] {\large $c$};
		\node[inner sep=2pt,minimum size=0pt] (CC) at (1.5,-3) [circle,draw] {};
		\node at (CC) [left,xshift=-1mm] {\large $c_c$};
		\node[fill,inner sep=2pt,minimum size=0pt] (D) at (3,-3) [circle,draw] {};
		\node at (D) [left,xshift=-1mm] {\large $d$};

		\draw[dashed] (R) -- (CR);
		\draw[dashed] (R) -- (X2);
		\draw (X2) -- (A);
		\draw[dashed] (X2) -- (X1);
		\draw (X1) -- (B);
		\draw (X1) -- (C);
		\draw[dashed] (C) -- (CC);
		\draw (C) -- (D);
		
	\end{tikzpicture}
	\caption{After binarizing the tree.}
	\label{fig:tree_dp_example_c}
	\end{subfigure}
	\caption{Example of the input preprocessing in the proof of Theorem \ref{thm:f_energy_poly_on_consistent_trees}. The hollow vertices are the newly added vertices and the dashed edges have weight~$0$ for both weight functions~$w_a$ and~$w_b$.}
	\label{fig:tree_dp_example}
\end{figure*}
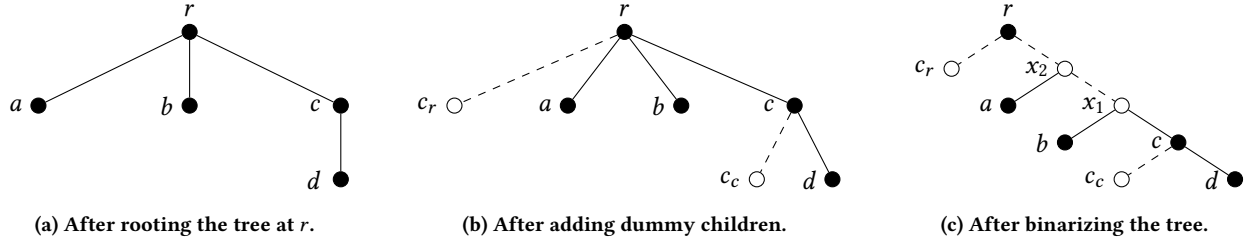

Motivated by the hardness on general graphs, we now consider the two energy objective functions on trees. We start by showing that~$f_{\text{energy}}$ can be solved in polynomial time on trees if we have consistent agent weights.

\begin{theorem}\label{thm:f_energy_poly_on_consistent_trees}
  \mbsp{} with objective function~$f_{\text{energy}}$ can be solved in~$\Oh(n^4 k^2 + |A|)$ time on trees with consistent agent weights.
\end{theorem}

\begin{proof}
  We describe a dynamic programming algorithm. First, for~$f_{\text{energy}}$ and consistent weights, we can simplify the problem by turning each agent~$a = (s_a, t_a, w_a)$ into two agents~$a_s = s_a$ and~$a_t = t_a$ that both only have a starting point. This is because the agent always needs to enter the bus and it does not matter how far the agent travels in the bus. We write~$A_v$ to denote the set of agents that have~$v$ as their start vertex. We can precalculate all of these sets in time~$\Oh(n + |A|)$. For the algorithm to work we modify the input tree so that it becomes a rooted and binary tree where we only need to place bus stops at leaves:	

	Let~$T = (V, E)$ be the input tree. We start by rooting~$T$ at an arbitrary vertex~$r$ as shown in Figure~\ref{fig:tree_dp_example_a}. For each vertex~$v\in V$ with at least one child in~$T$ we add a new dummy child~$c_v$, and we set the weight of the edge~$\{ v, c_v \}$ to zero for both~$w_a$ and~$w_b$. An example of this is shown in Figure~\ref{fig:tree_dp_example_b}. Notice that placing a bus stop at~$v$ is now equivalent to placing a bus stop at~$c_v$. This already ensures that at least one optimal solution exists where all bus stops are at leaves. Next, we binarize the tree. Let~$v$ be a vertex with at least three children~$v_1, v_2, v_3$. We add a new child~$x$ to~$v$ and change~$v_2$ and~$v_3$ such that they are now children of~$x$ instead of~$v$. We set the weight of the edge~$\{ v, x \}$ to zero for both~$w_a$ and~$w_b$ and we set the weights of the edges~$\{x, v_2\}$ and~$\{x, v_3\}$ to the weights of the old edges~$\{v, v_2\}$ and~$\{v, v_3\}$, respectively.
	In Figure~\ref{fig:tree_dp_example_c} two extra vertices were added since the vertex~$r$ had four children. Note that these changes only lead to a linear increase in the size of the tree.

  Let~$T$ be the final binarized tree. We write~$L$ to denote the set of leaf vertices in~$T$ and~$L_v$ to denote the set of leaf vertices in the subtree with root~$v$.
  We also order the tree, meaning each non-leaf vertex~$v$ has a fixed left and right child.
  As a final preparation we calculate all pairwise distances in the tree for~$w_a$.

  We now define the arguments of the dynamic programming table~$D[v, k', b_{in}, b_{out}, p]$ as follows:
  \begin{itemize}
    \item $v$: the root of the subtree we are currently looking at.
    \item $k'$: the number of bus stops in the subtree of~$v$.
    \item $b_{in}$: a bus stop in the subtree of~$v$ where~$\pi(v, b_{in}, w_a)$ is minimal, it is~$\bot$ if there is no such bus stop.
    \item $b_{out}$: a bus stop outside the subtree of~$v$ where~$\pi(v, b_{out}, w_a)$ is minimal, it is~$\bot$ if there is no such bus stop.
    \item $p$: the number of stops in the subtree of~$v$ where the overall bus route starts or ends, can either be~$0$,~$1$, or~$2$.
  \end{itemize}
	
  Each entry in the table is the sum of two parts. First, for all agents in the subtree of~$v$, the sum of their distances to their closest bus stop (even if this bus stop is outside the subtree). Second, the sum over the bus weights of each edge $e$ in the subtree of~$v$ multiplied by how often the bus traverses $e$.  To calculate the first part we may need to calculate the distance~$\pi(v,b,w_a)$ between a vertex~\(v\) and a bus stop~$b$. For simplicity, if~$b = \bot$ we set~$\pi(v,b,w_a) = \infty$.

  To calculate the table entries~$D[\ell, k', b_{in}, b_{out}, p]$ for a leaf~$\ell$ we only need to consider two cases. If~$\ell$ is a bus stop, then the agents that start at~$\ell$ do not need to walk anywhere. If~$\ell$ is not a bus stop, then the agents that start at~$\ell$ walk to~$b_{out}$. The second part is always~$0$ since leaves do not have any children:
\begin{equation*}
\begin{cases}
      0                                        & \text{if } k' = 1, b_{in} = \ell, p \leq k',                 \\
      |A_{\ell}| \cdot \pi(\ell, b_{out}, w_a) & \text{if } k' = 0, b_{in} = \bot, b_{out} \not= \bot, p = 0, \\
      \infty                                   & \text{otherwise}.
    \end{cases}
\end{equation*}

  We now look at a non-leaf vertex~$v$ with left child~$\ell$ and right child~$r$. We can immediately set the table entry to~$\infty$ if~$k' < p$ or~$k' = 0$ but~$b_{in} \not= \bot$ since there is no feasible solution with these properties. Otherwise, we use the following recurrence to calculate the table entry~$D[v, k', b_{in}, b_{out}, p]$. We will explain each part of it in the following paragraphs:
\begin{equation*}
    \min\limits_{\substack{
    k^{\ell} \in \{0, \hdots, k'\} \\
    (b^{\ell}_{in}, b^{\ell}_{out}, b^r_{in}, b^r_{out}) \in valid(\cdot) \\
    p^{\ell} \in \{0, \hdots,  p\}
    }}
    \left(
    \begin{aligned}
          & D[v^{\ell}, k^{\ell}, b^{\ell}_{in}, b^{\ell}_{out}, p^{\ell}] \\
        + & D[v^r, k' - k^{\ell}, b^r_{in}, b^r_{out}, p - p^{\ell}]       \\
        + & |A_v| \cdot \min_{b\in \{ b_{out}, b_{in} \}} \pi(v, b, w_a)   \\
        + & \delta(v, k^{\ell}, k' - k^{\ell}, p^{\ell}, p - p^{\ell})
      \end{aligned}
    \right)
\end{equation*}

The general idea is to go through all combinations of table entries for the left and right child of~$v$ that fulfill the constraints of the table entry we are trying to calculate. Once we have fixed the table entries of the two children, we can calculate where the agents that start at~$v$ walk to and how many times the bus traverses the edges between~$v$ and its two children.

To check all relevant table entries of the two children, we first go through all possible values of~$k^\ell$ and~$p^\ell$, thereby distributing the~$k'$ bus stops and the~$p$ end stops between the two child subtrees. Since~$k' \leq k$ and~$p \leq 2$, there are only~$\Oh(k)$ ways to do this. Next, we go through all tuples~$(b^{\ell}_{in}, b^{\ell}_{out}, b^r_{in}, b^r_{out})$ of relevant~$b_{in}$ and~$b_{out}$ bus stops for the two children. We call the set of these tuples~$valid(\cdot)$. In the following, we show that there are only~$\Oh(n)$ such tuples.

  	Without loss of generality, assume that~$b_{in}$ is in the subtree of~$\ell$; the case where it lies in the subtree of~$r$ is symmetric. We can now set~$b^\ell_{in} = b_{in}$ for all tuples since any bus stop that would be closer to~$\ell$ would also be closer to~$v$. For~$b^r_{in}$ we can choose any vertex that is in the subtree of~$r$ and is not closer to~$v$ than~$b_{in}$, meaning~$\pi(b^r_{in}, v, w_a) \geq \pi(b_{in}, v, w_a)$. Additionally, we can choose~$b^r_{in} = \bot$ for the case that there is no bus stop in the subtree of~$r$. Once we have fixed~$b^\ell_{in}$ and~$b^r_{in}$ we can uniquely determine~$b^\ell_{out}$ and~$b^r_{out}$ as follows:
  	\begin{align*}
    	b^{\ell}_{out} = \argmin_{b\in \{ b^r_{in}, b_{out} \}} \pi(\ell, b, w_a) &&
    	b^r_{out} = \argmin_{b\in \{ b^{\ell}_{in}, b_{out} \}} \pi(r, b, w_a)
  	\end{align*}
  	In total this shows that~$|valid(\cdot)| \in \Oh(n)$.

  Once the two child entries are fixed, both parts of the current entry can be computed in constant time. The agents starting at~$v$ walk either to~$b_{in}$ or to~$b_{out}$, so their contribution is obtained from the precomputed distances and~$|A_v|$.
  
  It remains to account for the two edges from~$v$ to its children. We denote this contribution by~$\delta(v,k^\ell,k^r,p^\ell,p^r)$. For the edge~$\{v,\ell\}$, the contribution is~$0$ if~$k^\ell=0$ or~$k^\ell=k$. Otherwise, the bus traverses this edge once if~$p^\ell=1$, and twice if~$p^\ell\in\{0,2\}$. The contribution of~$\{v,r\}$ is defined analogously, and~$\delta(v,k^\ell,k^r,p^\ell,p^r)$ is the sum of the two contributions.
  
  Overall, this gives a running time of~$\Oh(nk)$ for calculating each of the~$\Oh(n^3k)$ table entries. Combined with the time to precalculate all pairwise distances and the values~$|A_v|$, this gives a total running time of~$\Oh(n^4k^2+|A|)$.
\end{proof}

In contrast to Theorem~\ref{thm:f_energy_poly_on_consistent_trees} we show that allowing arbitrary agent weight functions makes the problem W[2]-hard even if we further restrict the input graphs to stars. For this we again use a reduction from \textsc{Hitting Set}.

Given a \textsc{Hitting Set} instance~\((\hg, k)\), we construct an \mbsp{} instance on a star~\(G\). Let~\(c \in V(G)\) be the center vertex of the star.
For each vertex~\(v \in V(\hg)\), we add a leaf vertex~\(x_v\) to~\(V(G)\).
Then we add a leaf vertex~\(x_e\) to~\(V(G)\) for each edge~\(e \in \he(\hg)\).
We set all bus-weights to zero.
For each~\(v \in V(G)\) we add a \emph{vertex-agent}~\(a_v\) with start and end vertex~\(x_v\).
These agents have an edge weight of 1 on the edge~\(\{ x_v, c\}\) and~\(0\) on all other edges.
For each~\( e\in \he(G) \) we add an \emph{edge-agent}~\(a_e\) with start and end vertex~\(x_e\).
These agents have weight~\(0\) on~\(\{x_e, c\}\) and~\(\{c, x_v\}\) for all~\(v \in e\) and weight~$1$ on all other edges. The parameter~$k$ is the same.

\begin{theorem}\label{thm:f_W2_on_arbitrarily_weighted_stars}
	\mbsp{} with objective function~\(f_{\text{energy}}\) or~\(f_{\text{time}}\) is W[2]-hard for parameter~\(k\) on stars with arbitrary agent weights.
\end{theorem}
\begin{proof}
	Since we set the bus weights to zero, the objective functions~\(f_{\text{energy}}\) and~\(f_{\text{time}}\) coincide, and we refer to them as~\(f\).
	We can prove the statement by showing the following claim:

	\emph{Claim}: There exists a hitting set of size~\(k\) in~\(\hg\) if and only if the \mbsp{} instance has a solution~$B$ with~\( f(B) \leq 2(n-k)\).

	\((\Rightarrow)\) Let~\(S\) be a hitting set of size~\(k\) in~\(\hg\).
	Then by placing a bus stop at each vertex~\(x_v\) for~\(v \in S\), we obtain a solution where we have~\(n-k\) vertex agents contributing a cost of~\(2\) to the objective function value each.
	By construction of the \mbsp{} instance on~\(G\), each edge-agent~\(a_e\) has weight~\(0\) on all walks from~\(x_e\) to a vertex~\(x_v\) for~\(v \in e\). By definition of~\(S\) for each~\(e \in \he{}\) we have~\(S \cap e \neq \emptyset\) and hence~\(a_e\) can reach a bus stop at~\(0\) cost from~\(x_e\).
	Subsequently, edge-agents do not contribute to the objective function value.

	\((\Leftarrow)\) Consider a solution to the \mbsp{} instance with objective value at most~\(2(n-k)\).
	Because each vertex-agent which is not at a bus stop contributes~\(2\) to the objective function value, each of the~\(k\) bus stops must be placed at a leaf vertex~\(x_v\) for some~\(v \in V(\hg)\).
	Since all other vertex-agents contribute~\(2\) to~\(f\), the edge-agents must contribute~\(0\).
	Thus, for all~\(e \in \he\) there must be a weight-zero walk from~\(x_e\) to a bus stop~\(x_v\) for some vertex~\(v \in e\).
	Therefore,~\(S \coloneqq \{v \in V(\hg) \mid x_v \text{ is bus stop}\}\) is a hitting set of size~\(k\).
\end{proof}

Next, we look at the cases covered by Theorems~\ref{thm:f_energy_poly_on_consistent_trees} and~\ref{thm:f_W2_on_arbitrarily_weighted_stars} for the~$g_{\text{energy}}$ function. Surprisingly, allowing agents to directly walk to their destination makes all of these cases W[1]-hard. We show this with a reduction from \textsc{Clique} which is W[1]-hard for~$k$~\cite{DF13}.

\begin{theorem}\label{thm:g_energy_hard_on_unweighted_stars}
	\mbsp{} with objective function~$g_{\text{energy}}$ is W[1]-hard on stars in the unweighted case.
\end{theorem}

\begin{proof}
	We use a reduction from the \textsc{Clique} problem, where we are given a graph~$G$ and an integer~$k$, and we want to know if~$G$ contains a clique of size~$k$. Our reduction works as follows:
	
	Let~$(G,k)$ be an instance of \textsc{Clique}. We create a new graph~$T$.
	First, we add a center vertex~$c$. For each vertex~$v\in V(G)$ we add the two vertices~$s_v$ and~$t_v$ to~$T$. Since we are creating a star graph, all of these vertices are directly connected to the center vertex~$c$. 
	
	Let~$M > 4k + 2|E(G)|$. For each vertex~$v\in V(G)$ we add~$M$ \emph{vertex agents} with start vertex~$s_v$ and end vertex~$t_v$. We call the set of these agents~$A_v$. For each edge~$e = \{ u, v \}\in E(G)$ we add one \emph{edge agent}~$a_e$ with start vertex~$s_u$ and end vertex~$s_v$. Finally, we set the number of allowed bus stops to~$k' = 2k$. This gives us the \mbsp{} instance~$(T, A, k')$. Note, that we can ignore the weight functions since we are in the unweighted case.
	
	We now show that~$(G,k)$ is a yes-instance if and only if~$(T, A, k')$ has a bus route~$B$ with~$g_{\text{energy}}(B) \leq 4k - 2 + 2|E(G)| - 2\binom{k}{2} + 2M\cdot (|V(G)| - k)$.
	
	($\Rightarrow$) Let~$C$ be the vertex set of a clique in~$G$ with~$|C| = k$. For each vertex~$v\in C$ we add a bus stop at~$s_v$ and a bus stop at~$t_v$. We put all bus stops into an arbitrary order to create a bus route~$B$. The bus route itself contributes~$4k - 2$ to the value of the objective function. This is because we have to first drive from the first bus stop to the center vertex. Then we have to visit the next~$2k - 2$ bus stops by driving to them and then back to the center vertex. Finally, we drive from the center vertex to the last bus stop.
	
	For each vertex~$v\in V(G)$ the agents~$A_v$ contribute~$2M$ to the objective function if~$v \notin C$ since each agent needs to walk two edges from~$s_v$ to~$t_v$. Otherwise, the agents can take the bus and have a contribution of zero. This gives us a total contribution of~$2M\cdot (|V(G)| - k)$. For each edge~$e = \{u, v\}\in E(G)$ the agent~$a_e$ contributes~$2$ to the objective function if at most one of~$u$ and~$v$ is a bus stop, and it contributes zero if both vertices are bus stops. Since we selected the vertices of a clique of size~$k$, these agents have a total contribution of~$2|E(G)| - 2\binom{k}{2}$.
	
	($\Leftarrow$) Let~$B$ be a bus route with~$g_{\text{energy}}(B) \leq 4k - 2 + 2|E(G)| - 2\binom{k}{2} + 2M\cdot (|V(G)| - k)$. We start by showing that the center vertex~$c$ cannot be a bus stop. For this, let us assume that the center vertex is a bus stop. We will look at how much the vertex agents contribute to the objective function under this assumption. Let~$v\in V(G)$. If neither~$s_v$ nor~$t_v$ are bus stops, then the agents in~$A_v$ have a total contribution of~$2M$. If exactly one of these vertices is a bus stop then the agents have a total contribution of~$M$. And if both are bus stops, then the contribution is zero. This means every non-center vertex in~$T$ that becomes a bus stop reduces the contribution of the vertex agents by~$M$. Since the center vertex is already a bus stop, we can only select at most~$2k - 1$ of these vertices for a total contribution of~$2M\cdot |V(G)| - M\cdot (2k - 1) = M + 2M\cdot (|V(G)| - k)$. By the definition of~$M$ this is bigger than~$4k - 2 + 2|E(G)| - 2\binom{k}{2} + 2M\cdot (|V(G)| - k)$. Hence, the center vertex cannot be a bus stop.
	
	Next, we show that there must be exactly~$k$ vertices in~$V(G)$ for which both~$s_v$ and~$t_v$ are bus stops. For this, we again look at the contribution of the vertex agents. Let~$v\in V(G)$. If at most one of~$s_v$ and~$t_v$ is a bus stop, then the agents in~$A_v$ contribute~$2M$ to the objective function since the center vertex is not a bus stop, meaning they have to always walk at least two edges. If both vertices are bus stops, then the contribution is zero. This means, the only way to reduce the contribution of the vertex agents is to select both their starting point and their end point. We can do this at most~$k$ times which gives us a total contribution of~$2M\cdot (|V(G)| - k)$. Doing this only~$k-1$ times would lead to a contribution that is bigger than~$4k - 2 + 2|E(G)| - 2\binom{k}{2} + 2M\cdot (|V(G)| - k)$. Hence, there must be exactly~$k$ vertices in~$V(G)$ for which both~$s_v$ and~$t_v$ are bus stops. We will call these~$k$ vertices the vertices that were selected by~$B$.
	
	The only thing left to show is that the vertices selected by~$B$ must form a clique in~$G$. For this, notice that the contribution of the bus route and the vertex agents is now fixed to~$4k - 2$ and~$2M\cdot (|V(G)| - k)$, respectively. This means that the remaining edge agents must have a contribution of at most~$2|E(G)| - 2\binom{k}{2}$. As we explained in the~$(\Rightarrow)$ part of the proof, the contribution of an edge agent~$a_e$ can only be reduced if both end points of~$e$ are selected by the bus route. This means the minimum contribution is reached if the bus route selects a clique. We already saw in the~$(\Rightarrow)$ part that this leads to a contribution of~$2|E(G)| - 2\binom{k}{2}$ which shows that the bus route must select a clique of size~$k$ in~$G$.
\end{proof}

Finally, we restrict the input graphs to paths and show that~\mbsp{} can be solved efficiently for both energy objectives, even with arbitrary agent weights. As in Theorem~\ref{thm:f_energy_poly_on_consistent_trees}, our approach is based on dynamic programming. The idea is to fix the rightmost bus stop so far and consider all possibilities for the next bus stop. This determines how agents traverse the edges between the two stops.

\begin{theorem}\label{thm:energy_poly_on_arbitrary_paths}
	\mbsp{} with objective function~$f_{\text{energy}}$ or~$g_{\text{energy}}$ can be solved in~$\Oh(n^3k\cdot |A|)$ time on paths with arbitrary weights.
\end{theorem}
\begin{proof}
We first describe the algorithm for~$g_{\text{energy}}$ and then show how it can be modified to work for~$f_{\text{energy}}$.
Let~$(P, A, w_b, k)$ be the input. The graph~$P = (V,E)$ is a path on~$n$ vertices~$(v_1, \hdots, v_n)$.
We say that~$v_i$ is to the left of~$v_j$ and that~$v_j$ is to the right of~$v_i$ if~$i < j$.
We assume that each agent in~$A$ has start vertex~$v_\ell$ and end vertex~$v_r$ with~$\ell \leq r$.

We define a dynamic programming based algorithm to solve \mbsp{} for this input.
A table entry~$D[k',v_j]$ assumes that~$v_j$ is a bus stop and that there are exactly~$k'$ bus stops to the left of~$v_j$. The entry contains the objective function value of an optimal bus route that fulfills these conditions while only counting the contribution of edges that are to the left of~$v_j$.
In the base cases with~$k' = 0$ the bus route contributes~$0$ to the objective function. This means we only need to count the contribution of the agents. If an agent starts and ends to the left of~$v_j$, then it is optimal for that agent to walk directly to their endpoint. Else if the agent only starts to the left of~$v_j$, then that agent will always walk to~$v_j$. If the agent does not start to the left of~$v_j$, then the agent has a contribution of~$0$ since it does not make sense for the agent to walk to the left of~$v_j$:
\begin{equation*}
	D[0, v_j] = \sum_{(v_\ell, v_r, w)\in A}
    \begin{cases}
      \pi(v_\ell, v_r, w) & \text{if } r < j,           \\
      \pi(v_\ell, v_j, w) & \text{if } \ell < j \leq r, \\
      0                   & \text{otherwise}.
    \end{cases}
\end{equation*}
If~$k' > 0$, then we need to check all possible positions~$v_i$ of the closest bus stop to the left of~$v_j$. For each of those we need to add up the previous table entry~$D[k' - 1, v_i]$, the contribution of the bus between~$v_i$ and~$v_j$, and the contribution of the agents between~$v_i$ and~$v_j$. This can be expressed as follows, where~$f_A(v_i, v_j)$ is the agent contribution between~$v_i$ and~$v_j$.
\begin{equation*}
	\min_{i\in [j - 1]}
    \left( D[k' - 1, v_i] + \pi(v_i, v_j, w_b) + f_A(v_i, v_j) \right)
\end{equation*}

To calculate~$f_A(v_i, v_j)$ we need to sum up the contribution of each individual agent. For this we consider four different types of agents~$(v_\ell, v_r, w)\in A$.

\smallskip
\noindent
\emph{Case~$1$:~$i \leq \ell \leq r \leq j$.} In this case the agent starts and ends between~$v_i$ and~$v_j$. Here the agent either directly walks to their end point or they walk from their start point to~$v_i$, use the bus to get to~$v_j$, and then walk from~$v_j$ to their end point. The contribution of such an agent is~$\min(\pi(v_\ell, v_r, w), \pi(v_i, v_\ell, w) + \pi(v_r, v_j, w))$.

\smallskip
\noindent
\emph{Case~$2$:~$\ell < i \leq r \leq j$.} In this case only the end point is between~$v_i$ and~$v_j$. Here we can assume that the agent sits in the bus at~$v_i$, meaning they either walk from~$v_i$ to their end point or they use the bus to get to~$v_j$ and then walk from there. The contribution of such an agent is~$\min(\pi(v_i, v_r, w), \pi(v_r, v_j, w))$.

\smallskip
\noindent
\emph{Case~$3$:~$i \leq \ell \leq j < r$.} In this case only the start point is between~$v_i$ and~$v_j$. Analogously to Case~$2$, the agent either walks to~$v_i$ or~$v_j$. The contribution of such an agent is~$\min(\pi(v_i, v_\ell, w), \pi(v_\ell, v_j, w))$.

\smallskip
\noindent
\emph{Case~$4$:~$\ell < i \leq j < r$.} In this case neither start nor end point are between~$v_i$ and~$v_j$. Here we can assume that the agent uses the bus to get from~$v_i$ to~$v_j$, meaning the contribution is~$0$.

Finally, to find the globally optimal objective function value we need to consider each vertex as the rightmost bus stop. Similar to the case for~$k' = 0$ we also need to deal with agents that have their end point to the right of the rightmost bus stop. Overall we can calculate the optimal objective function value as follows:
\begin{equation*}
	\min\limits_{j\in [n]} \left( D[k,v_j] + \sum\limits_{(v_\ell, v_r, w)\in A}
    \begin{cases}
      \pi(v_\ell, v_r, w) & \text{if } j < \ell,        \\
      \pi(v_j, v_r, w)    & \text{if } \ell \leq j < r, \\
      0                   & \text{otherwise}.
    \end{cases}
    \right)
\end{equation*}

The table has size~$\Oh(nk)$, calculating the table entries for~$k' = 0$ takes~$\Oh(n \cdot |A|)$ time because of the shortest path calculation for each agent, and calculating the table entries for~$k' > 0$ takes~$\Oh(n^2 \cdot |A|)$ time because of the shortest path calculation for each agent and each next bus stop. In total this gives us a running time of~$\Oh(n^3k\cdot |A|)$.

For the algorithm to work for the~$f_{\text{energy}}$ objective function we only need to change two things. First, agents that start and end to the left of the leftmost bus stop or to the right of the rightmost bus stop can no longer walk directly to their end point and must instead walk to the leftmost or rightmost bus stop. Second, in the calculation of~$f_A(v_i, v_j)$ agents that have their start and end point between~$v_i$ and~$v_j$ can no longer walk directly to their end points. Instead, they have the option of either not using the bus and walking to one of the bus stops and back again, or they can walk to~$v_i$, use the bus to get to~$v_j$, and walk to their end point from there. None of these changes affect the asymptotic running time bound.
\end{proof}

\section{Minimizing Time on Trees}
\label{sec:time}

We now consider the time objectives on trees. We first show that~$f_{\text{time}}$ and~$g_{\text{time}}$ are W[1]-hard on stars, even with unit agent weights.
The key difficulty is that, for the time objectives, the effect of selecting one bus stop depends on which other stops are selected. We exploit this interaction to prove W[1]-hardness for this variant of \mbsp{}.

\begin{theorem}\label{thm:time_hard_on_unit_stars}
	\mbsp{} with objective function~$f_{\text{time}}$ or~$g_{\text{time}}$ is W[1]-hard on stars with unit agent weights.
\end{theorem}
\begin{proof}
We use a reduction from \textsc{Independent Set} to \mbsp{}.
Let~$(G=(V,E),k)$ be an instance of \textsc{Independent Set} where~$\abs{V} = n > \max\{4, k\}$.

The following proofs use the variables~$\alpha$,~$\beta$ and~$\gamma$. These variables are chosen such that they are separated by orders of magnitude.
We define~$\alpha, \beta$ and~$\gamma$ as follows:
\begin{align*}
	\gamma &{}= 4 \cdot n^2, \ \ \ & \beta &{}= 3n^3 \gamma, \ \ \ & \alpha &{}= 2n \cdot \beta + 2 n^2 \cdot \gamma + 4 n^2.
\end{align*}
We construct an instance~$(G',A^*,w_b,w_a,3k)$ of \mbsp{} as follows:
\begin{align*}
	V(G') ={}& \{c\} \cup V_1 \cup V_2 \\
	V_1 ={}& \{X_v \mid v \in V\} \\
	V_2 ={}& \{X_i \mid i \in [k]\} \\
	E(G') ={}& \{\{c, v'\}, \mid v' \in V(G') \setminus \{c\}\} \\
	w_b(e) ={}& \alpha, \ \ \ w_a(e)=\alpha+1 \text{ for } e \in E(G') \\
    A^* ={}& A_N \cup A_A \cup A_E \\
    A_N ={}& \{a_{v',i} = (c,v', w_a) \mid v' \in V(G') \setminus \{c\},i \in [\beta]\} \\
    A_A ={}& \{a_{v,j,i} = (X_v, X_j, w_a) \mid v \in V, j \in [k], i \in [\gamma]\} \\
    A_E ={}& \{a_{u,v} = (X_v, X_u, w_a) \mid \{u,v\} \in E\} \\
	&{}\cup \{a_{u,v,i} = (X_v,c,w_a) \mid \{u,v\} \notin E, u \neq v, \\ 
    &{} u,v \in V, i \in [2]\} 
\end{align*}

Thus,~$G'$ is a star where~$c$ is the center vertex which is connected to a vertex~$X_v \in V_1$ for~$v \in V$ and a vertex~$X_{i} \in V_2$ for each~$i \in [k]$.

The agents in $A_E$ represent the edges in $G$. The instance is constructed such that any bus route minimizing the objective function will have a specific structure. Any such route will be made up of~$k$ bus stops in~$V_1$,~$k$ bus stops in~$c$ and~$k$ bus stops in~$V_2$, always alternating between them in the same order. The only vertex that should be visited multiple times is~$c$ while every single one of the~$k$ vertices in~$V_2$ should be visited exactly once. If an independent set exists, the $k$ vertices $X_v \in V_1$ in the route represent $k$ vertices in $G$ which form such an independent set.

As we show in \Cref{agent_shortest_path}, the agent and bus weights are so similar, that it never pays off to travel more edges than necessary.
For a bus route~$B = (b_1, \dots , b_{3k})$, we say bus stops~$b_i,b_j \in B$ are \emph{neighboring} if the distance between $b_i$ and $b_j$ in $G'$ is $1$ and~$\abs{i-j} = 1$.
For two pairs~$b_i,b_j \in B$ and~$b_l,b_m \in B$, we say they are \emph{distinct} if neither~$b_i = b_l$ and~$b_j = b_m$ nor~$b_i = b_m$ and~$b_j = b_l$.
The set~$A_N$ is constructed to maximize the number of neighboring bus stops. We call the agents in $A_N$ \emph{neighbor agents}. We call the agents in $A_A$ \emph{alternation agents} as the set~$A_A$ is constructed to ensure that an optimal bus route alternates between vertices in~$V_1$ and~$V_2$. The set~$A_E$ is constructed to ensure the vertices of \(V_1\) in the route represent an independent set in \(G\).

The total numbers of agents by set are $\abs{A_N} = (n + k) \cdot \beta, \abs{A_A} = n \cdot k \cdot \gamma$ and $\abs{A_E} = 4 \cdot \binom{n}{2} - 2 \cdot \abs{E}$. The distance by agent is~$1$ for $A_N$ and~$2$ for $A_A$ while~$A_E$ contains agents with paths of length~$1$ and $2$. This results in the sum of the distance between the endpoints for each agent in~$A^*$ being $\Delta = (n + k) \cdot \beta + 2 n \cdot k \cdot \gamma + 4 \cdot \binom{n}{2}$. The total distance agents travel by bus in the desired solution is $\Theta = 2k \cdot \beta + (2k - 1 + n \cdot k) \cdot \gamma + k \cdot 2(n-1)$.

\begin{lemma}
	The independent set instance has a solution if and only if the objective function~$f_{\text{time}}$ for this instance is at most~$\Sigma = \alpha \cdot \Delta + \Delta - \Theta$.
    \label{objective_function}
\end{lemma}

We call a bus route~$B$ which fulfills~$f_{\text{time}} \leq \Sigma$ a \emph{certifying solution}.

To prove \Cref{objective_function}, we will first lay some groundwork.

\begin{lemma}
	In every certifying solution the total number of edges travelled by agents is $\Delta$.
	\label{agent_shortest_path}
\end{lemma}

\begin{proof}
\renewcommand{\qedsymbol}{$\diamond$}
	The sum of the distance between the endpoints for all agents is~$\Delta$ thus $\Delta$ is the minimum amount of edges the agents can travel in any solution.Since \(\alpha > \Delta\) it holds, that \(\Sigma < \alpha \cdot (\Delta+1)\) and more than $\Delta$ edges being travelled by agents in~$B$ would result in~$f_{\text{time}}(B) \geq (\Delta+1) \cdot \alpha > \Sigma$.
\end{proof}

For the following lemma we introduce some further vocabulary, for a bus route~$B$, we say a bus stop~$b_l$ is \emph{between~$b_i$ and~$b_j$} if~$i < l < j$ or~$j < l < i$, and we further say two vertices~$u',v' \in V(G') \setminus \{c\}$ are \emph{connected by~$B$} if there exist~$b_i = u', b_j = v'$ such that all bus stops between~$b_i$ and~$b_j$ are~$c$.

\begin{lemma}
	If there exists an independent set~$X$ of~$G$ with~$\abs{X}=k$, then there exists a certifying solution~$B$. \label{correct_route_exists}
\end{lemma}

\begin{proof}
\renewcommand{\qedsymbol}{$\diamond$}
	Let~$\vec{X} = (x_1, \dots , x_k)$ be any permutation of~$X$.
	The following bus route~$B$ satisfies the lemma.
	\begin{align*}
		b_{3i-2} ={}& X_{x_i} \in V_1 {}&\text{ for } i \in [k] \\
		b_{3i-1} ={}& c {}&\text{ for } i \in [k] \\
		b_{3i} ={}& X_{i} \in V_2 {} & \text{ for } i \in [k]
	\end{align*}
	Notice that every bus stop that is not~$c$ has~$c$ as a neighboring bus stop. Furthermore, every agent has~$c$ on its shortest path meaning that no agent has to take a detour to visit a bus stop. Since~$x_i \neq x_j$ for~$i \neq j$ and~$v_e \notin X$, it follows that~$b_i \neq b_j$, for~$i \neq j, b_i \neq c, b_j \neq c$.

    We will now show for how many edges agents can use the bus by differentiating the types of agents.

   	$A_N$: There are exactly~$2k$ distinct pairs of neighboring bus stops in~$B$. There are exactly~$\beta$ neighbor agents in~$A_N$ for every edge~$e' \in E(G')$. Thus,~$2k \cdot \beta$ neighbor agents~$a_{v',i} \in A_N$ use the bus for their one station.
	
   	$A_E$: Since~$X$ is an independent set of~$G$ and since every bus stop other than~$c$ neighbors~$c$, every agent in~$A_E$ with an endpoint~$b_i \in B, b_i \neq c$ uses the bus for the edge~$\{b_i,c\}$. Since~$2$ such agents exist for every pair of vertices~$b_i, v' \in V_1$ there are~$2(n-1)$ such agents for every bus stop~$b_i \in V_1$. Since~$k$ vertices in~$V_1$ are visited by~$B$, this represents a total of~$k \cdot 2(n-1)$ edges travelled by bus for the agents in~$A_E$.
	
    $A_A:$ Since every vertex in $V_2$ has $c$ as a neighboring bus stop every alternation agent can use the bus between its endpoint in $V_2$ and $c$ for a length of $1$. 
	Furthermore, alternation agents whose endpoints are connected by $B$ use the bus for both edges. This is the case for $(2k-1)$ such pairs of bus stops each with $\gamma$ corresponding agents in \(A_A\) resulting in a total of~$(n \cdot k + 2k-1)  \cdot \gamma$ edges travelled by bus by alternation agents.

    If every agent travels its shortest path and at least $\Theta = 2k \cdot \beta + k \cdot 2 (n-1) + 2k \cdot (n-1) \cdot \gamma$ edges among them are travelled by bus, then $f_{\text{time}}(B) \leq \alpha \cdot \Theta + (\alpha+1) \cdot (\Delta - \Theta) = \Sigma$. 
	\label{proof_Rightarrow}
\end{proof}

\begin{lemma}
	A certifying solution contains exactly~$2k$ distinct pairs of neighboring bus stops along the route. \label{2k_pairs_exact}
\end{lemma}

\begin{proof}
\renewcommand{\qedsymbol}{$\diamond$}
	We first show that no certifying solution contains less than~$2k$ such pairs. Every neighbor agent~$a_{v',i} \in A_N$ uses the bus if and only if there exist neighboring bus stops along the bus route between~$X_{v'}$ and~$c$. Less than~$2k$ distinct pairs of neighboring bus stops means at least~$(\abs{V(G')}-2k) \cdot \beta$ neighbor agents taking the bus. Thus, any such bus route cannot achieve an objective function any lower than $\alpha \cdot (\Delta - (\abs{V(G')}-2k) \cdot \beta) + (\alpha + 1) \cdot (\abs{V(G')}-2k) \cdot \beta$.
	\begin{align*}
		(\Delta - \Theta) ={}& (\abs{V(G')}-1) \cdot \beta + n \cdot (n-1) \cdot 2 \cdot \gamma \\&{}+ \binom{n}{2} \cdot (n-3) \cdot \gamma + 4 \cdot \binom{n}{2} - \Theta \\
		<{}& (\abs{V(G')}-1) \cdot \beta + 3n^3 \cdot \gamma - 2k \cdot \beta \\
		={}& (\abs{V(G')}-2k) \cdot \beta
	\end{align*}
	Since~$(\abs{V(G')}-2k) \cdot \beta > \Delta - \Theta$, and \(\Delta - \Gamma\) is an upper bound for the number of edges travelled by bus in a certifying solution, no solution with~$2k-1$ or fewer distinct pairs of neighboring bus stops is certifying.

	To prove that no route can contain more than~$2k$ distinct pairs of neighboring bus stops, we show that there are never~\(3\) consecutive such pairs. Toward contradiction, let~$b_i$ be a bus stop such that the pairs~$(b_i, b_{i+1}),$ $(b_{i+1}, b_{i+2})$ and~$(b_{i+2}, b_{i+3})$ are all distinct and neighboring for~$i \in [3k-3]$. Since~$b_{i+1}$ and~$b_{i+2}$ are neighboring one of them is~$c$. Without loss of generality~$b_{i+1} = c$. Since~$c$ is the only neighbor of of~$b_{i+2}$ in~$G'$, it follows that~$c = b_{i+3}$. Thus, the pairs~$(b_{i+1},b_{i+2})$ and~$(b_{i+2},b_{i+3})$ are not distinct; a contradiction.

	Since every certifying solution contains at least~$2k$ distinct pairs of neighboring bus stops along the route but never more than~$2$ consecutive pairs of neighboring bus stops, a certifying solution contains exactly~$k$ sequences of~$2$ consecutive instances of distinct neighboring bus stops along the route resulting in exactly~$2k$ distinct pairs of neighboring bus stops.
\end{proof}

\begin{lemma}
	A certifying solution~$B$ contains exactly~$2k-1$ distinct pairs of connected bus stops~$b_i,b_j$ such that~$\{b_i,b_j\} \cap V_1$ and~$\{b_i,b_j\} \cap V_2$ are non-empty.
	\label{2kplus1connected}
\end{lemma}

\begin{proof}
\renewcommand{\qedsymbol}{$\diamond$}
	By \Cref{2k_pairs_exact}, a certifying solution contains exactly~$k$ sequences of~$2$ consecutive instances of distinct neighboring bus stops along the route. All such sequences start and end with a vertex in~$V_1 \cup V_2$ while the middle bus stop of such a sequence is~$c$. Thus, a certifying solution contains~$c$ exactly~$k$ times. This means that, for a certifying bus route~$B$, there are exactly~$2k-1$ distinct pairs of connected vertices~$u',v' \in V_1 \cup V_2$.
	Assume towards contradiction that at most~$2k-2$ of those pairs are between~$V_1$ and~$V_2$. Then at  most~$(2k-2) \cdot \gamma$ agents in \(A_N\) can travel both edges by bus. Since at most $2k \cdot \beta$ edges are travelled by bus by neighbor agents even if all edges for agents in $A_E$ are travelled by bus this leaves at most $2k \cdot \beta + (2k - 2 + n \cdot k) \cdot \gamma + 4 \cdot \binom{n}{2}$ edges to be travelled by bus in~$B$.
	\begin{align*}
		\Theta ={}& 2k \cdot \beta + (2k - 1 + n \cdot k) \cdot \gamma + 2k \cdot (n-1) \\
		>{}& 2k \cdot \beta + (2k - 1 + n \cdot k) \cdot \gamma \\
		>{}& 2k \cdot \beta + (2k - 2 + n \cdot k) \cdot \gamma + 4 \cdot \binom{n}{2}
	\end{align*}
	Since $2k \cdot \beta + (2k - 2 + n \cdot k) \cdot \gamma + 4 \cdot \binom{n}{2} < \Theta$, it follows that $f_{\text{time}}(B) > \Sigma$; a contradiction.
	\end{proof}
We will now prove \Cref{objective_function}.
\begin{proof}[Proof of \Cref{objective_function}.]
\renewcommand{\qedsymbol}{$\diamond$}
	\((\Rightarrow)\) By \Cref{correct_route_exists}, for any graph with an independent set of size at least~$k$, there exists a certifying solution.

	\((\Leftarrow)\) Let~$B$ be a certifying solution. To prove that we can construct an independent set of size~$k$ in~$G$ from~$B$, we will go through the properties of a certifying solution. By \Cref{2k_pairs_exact}, every certifying solution contains exactly~$2k$ distinct pairs of neighboring bus stops along the route. Thus, at most $2k \cdot \beta$ edges are travelled by bus by neighbor agents.
	By \Cref{2kplus1connected} there are exactly $2k-1$ distinct pairs of connected vertices between $V_1$ and $V_2$. First, this means that $B$ contains exactly $k$ vertices in $V_1$ and $V_2$ respectively. Second, since at most~$(2k-1) \cdot \gamma$ alternation agents can travel both edges by bus, even if every agent in $A_A$ travels at least one edge by bus at most~$(2k - 1 + n \cdot k) \cdot \gamma$ edges are travelled by bus by alternation agents.
	Since~$f_{\text{time}}(B) \leq \Sigma$ at least~$\Theta$ edges in~$B$ are travelled by bus. This leaves at least $\Theta - 2k \cdot \beta - (2k - 1 + n \cdot k) \cdot \gamma = k \cdot 2(n-1)$ edges to be travelled by bus by agents in $A_E$. Let $X = \{v \in V \mid X_{v}  \in B\}$ be the set of vertices in $G$ represented by the bus route. Since there are~$k$ bus stops in~$V_1$ and none of them are connected, every agent in $A_E$ travels at most one edge by bus and $\abs{X}=k$. For every bus stop $X_v \in V_1$, we define $A_v$ as the set of agents in~$A_E$ that have~$X_v$ as an endpoint. For every bus stop $X_v \in V_1$, there are exactly $2(n-1)$ agents in $A_E$ with one endpoint in $X_v$. As such $\abs{A_v} = 2(n-1)$ for all $X_v \in B, X_v \in V_1$. Agents in $A_E$ can only use the bus if one of their endpoints in $V_1$ is also in~$B$. Since agents in $A_E$ travel at least $k \cdot 2(n-1)$ edges by bus, every agent in~$A_E$ travels at most one edge by bus and every agent in $A_E$ that uses the bus is in some~$A_v$ it follows that $\abs{\bigcup_{v \in X} A_v} \geq k \cdot 2(n-1)$. Since $\sum_{v \in X} A_v = k \cdot 2(n-1)$ and $\abs{\bigcup_{v \in X} A_v} \geq k \cdot 2(n-1)$, this means that the sets~$A_v$ are pairwise disjoint.
	Since the sets $A_v$ are pairwise disjoint there is no agent in $A_E$ with endpoints $X_v, X_u \in B$.
	By the definition of $A_E$ this means that there exist no vertices $u,v \in X$ such that $\{u,v\} \in E$. Thus, $X$ is an independent set.
\end{proof}
Since $c$ is a bus stop in every bus route and between the endpoints for every agent, the proof above works for $g_{\text{time}}$ as well.
\end{proof}

Theorem \ref{thm:time_hard_on_unit_stars} and Proposition \ref{prop:g_time_poly_on_unweighted_graphs} cover all cases for stars and general trees except~$f_{\text{time}}$ in the unweighted case. Theorem~\ref{thm:f_time_poly_on_unweighted_stars} shows that this can be solved in polynomial time on stars while Theorem~\ref{thm:f_time_hard_on_unweighted_trees} shows that it is NP-hard on general trees.

\begin{theorem}\label{thm:f_time_poly_on_unweighted_stars}
	\mbsp{} with objective function~$f_{\text{time}}$ can be solved in polynomial time on stars in the unweighted case.
\end{theorem}
\begin{proof}
In the unweighted star setting, the order of the bus stops is irrelevant. Although every agent must use the bus, an agent may board and leave at the same stop. Hence, for any selected stop set, each agent can realize its best possible cost by choosing a suitable single stop from that set. Consequently, we regard a solution simply as a set of stops.
	
	Let~$(G,A,k)$ be an instance of \mbsp{} on an unweighted star~$G$ with center~$c$ and leaf set~$L$.
	For each leaf $v\in L$, let $n_v$ be the number of agents whose start and
	destination are both $v$. Sort the leaves such that $n_{v_1}\ge \cdots \ge n_{v_{|L|}}$ .
	Consider the two candidate solutions~$B_c \coloneqq \{c, v_1,\dots,v_{k-1}\}$ and~$B_\ell \coloneqq \{v_1,\dots,v_k\}$.
	We compute~$f_{\text{time}}(B_c)$ and~$f_{\text{time}}(B_\ell)$ and return the better of the two.
	
	We prove correctness. First consider solutions that contain the center $c$. Then every agent whose endpoints are distinct vertices can realize
	its shortest-path distance by using $c$ as its chosen stop. Thus the only
	part depending on the selected leaves is the contribution of agents with
	both endpoints at the same leaf. Such agents contribute $0$ if their leaf is
	selected and $2$ otherwise. Hence, among all solutions containing $c$, an
	optimal solution selects the $k-1$ leaves with largest values $n_v$, namely
	$B_c$.
	
	It remains to consider solutions not containing $c$. Let $S\subseteq L$ be
	a set of $k$ selected leaves. Let $x\in S$ be a selected leaf with minimum
	$n_x$. Compare $S$ with the center-containing solution
	$S' := (S\setminus\{x\})\cup\{c\}$.
	We analyze how the objective changes when~$x$ is replaced by~$c$.
	
	The agents located at $x$ increase their contribution from $0$ to $2n_x$.
	For every leaf $v\notin S$, the agents located at $v$ decrease their
	contribution from $4n_v$ to $2n_v$. Finally, every agent with distinct
	endpoints neither of which lies in $S$ decreases its contribution from $4$
	to $2$. Let $q(S)$ denote the number of such agents. It follows that
	$f_{\text{time}}(S') - f_{\text{time}}(S)
	=
	2n_x
	-
	2\sum_{v\notin S} n_v
	-
	2q(S)$.
	Consequently, if
	$n_x \le \sum_{v\notin S} n_v + q(S)$,
	then $S'$ is no worse than $S$.
	
	Now suppose that $S$ is strictly better than every solution containing
	$c$. Then the above inequality must be false, and therefore
	$n_x > \sum_{v\notin S} n_v + q(S)
	\ge
	\sum_{v\notin S} n_v$.
	In particular, every selected leaf has strictly larger $n_v$ than every
	unselected leaf. Hence $S$ must be exactly the set of the $k$ leaves with
	largest values $n_v$, that is, $S=B_\ell$ up to ties. If there is a tie
	between a selected and an unselected leaf, then $n_x\le \sum_{v\notin S}n_v$,
	making the center-containing solution optimal.
	
	Thus every optimal solution is either the best solution containing the center, namely $B_c$, or the solution $B_\ell$ using the $k$ leaves with
	largest values $n_v$. Both objective values can be computed in polynomial
	time, hence
	\mbsp{} with objective~$f_{\text{time}}$ is polynomial-time solvable on unweighted stars.
\end{proof}

\begin{theorem}\label{thm:f_time_hard_on_unweighted_trees}
	\mbsp{} with objective function~$f_{\text{time}}$ is NP-hard on trees in the unweighted case.
\end{theorem}
\begin{proof}
	We reduce from a \textsc{Vertex Cover} \cite{DBLP:conf/coco/Karp72} instance~$(G,k)$.
	We construct a tree~$T$ as follows.
	Add a root~$r$.  
	For every vertex~$v\in V(G)$, add two vertices~$v$ and~$v'$ to~$T$, connect~$v$ to~$r$, and connect~$v'$ to~$v$. This yields a tree of height two.
	
	We now define the agent set~$A$.  
	Choose integers~$M > 2|E(G)|$ and~$N > 2M\cdot (|V(G)| - k) + 2|E(G)|$.
	For every~$v\in V(G)$, add~$M$ agents whose start and end vertex is~$v$ (vertex agent set~$A_v$), and~$N$ agents whose start and end vertex is~$v'$ (prime agent set~$A_{v'}$).
        For every edge~$e=\{u,v\}\in E(G)$, add one agent~$a_e$ starting at~$u$ and ending at~$v$.  
	Set the allowed number of bus stops to~$k' = |V(G)| + k$.  
	This gives an instance~$(T,A,k')$ of \mbsp.
	
	We now show that~$(G,k)$ is a yes-instance if and only if~$(T, A, k')$ has a bus route~$B$ with~$f_{\text{time}}(B) \leq 2|E(G)| + 2M\cdot (|V(G)| - k)$.
	
	($\Rightarrow$) Let~$S$ be a vertex cover of~$G$ with~$|S| = k$. Place bus stops at all~$v'$ and at all~$v\in S$, using~$|V(G)| + k$ stops.  
	The order of stops is irrelevant for the objective, since for~$f_{\text{time}}$ an empty bus has no costs and in the unweighted case there is no difference between agents walking and agents taking the bus.
	
	For every~$v\notin S$, all agents in~$A_v$ must walk two edges to reach the corresponding prime bus stop and therefore contribute~$2M$.  
	For~$v\in S$ the contribution is~$0$. 
	Hence, the vertex agents contribute~$2M\cdot (|V(G)| - k)$ in total. Furthermore, since we place bus stops at all~$v'\in V(T)$ prime agents also have contribution~$0$.
	
	For every edge~$e=\{u,v\}$, at least one endpoint is a bus stop, so~$a_e$ can take its shortest path and incurs cost~$2$.
	Hence, the edge agents contribute~$2|E(G)|$.
	Thus,~$f_{\text{time}}(B) = 2|E(G)| + 2M\cdot (|V(G)| - k)$.
	
	($\Leftarrow$) Let~$B$ be a bus route with~$f_{\text{time}}(B) \leq 2|E(G)| + 2M\cdot (|V(G)| - k)$. 
	Since~$N > 2M\cdot (|V(G)| - k) + 2|E(G)|$, any~$v'$ not chosen as a bus stop would force each of its~$N$ agents to walk to a bus stop and back, thus leading to a cost of at least~$2$ per agent. By definition~$2N > 2M\cdot (|V(G)| - k) + 2|E(G)|$ exceeding the objective bound. Thus, all~$v'$ must be bus stops.
	Furthermore, the root~$r$ cannot be a bus stop.  
	Otherwise, only~$k-1$ vertices from~$V(G)$ could be chosen, and the vertex agents would contribute
	$2M\cdot |V(G)| - 2M(k-1)
		= 2M + 2M(|V(G)| - k)$,
	which is larger than the allowed objective value since~$M > 2|E(G)|$.
	
	Next, since all~$v'$ are bus stops, each~$A_v$ agent contributes~$2M$ unless~$v$ itself is a bus stop, in which case its contribution is~$0$.  
	To reach the required total of~$2M(|V(G)| - k)$, exactly~$k$ vertices of~$V(G)$ must be selected as bus stops.  
	Let this set be~$S_B$.
	The vertex-agent contribution is now fixed, so the edge agents must contribute at most~$2|E(G)|$.  
	An edge agent has cost~$2$ if at least one endpoint of its edge is a bus stop, and cost at least~$4$ otherwise (since~$r$ is not a bus stop and any detour increases distance).  
	Thus, to keep the cost at most~$2|E(G)|$, every edge must have at least one endpoint in~$S_B$. Thus,~$S_B$ is a vertex cover of size~$k$ and~$(G,k)$ is a yes-instance.
\end{proof}

Finally, we consider the case where the input graph is a path. 
Designing an efficient algorithm becomes challenging when agent weights are arbitrary. Intuitively, for an agent it may not be optimal to board and leave the bus at the stops closest to their start and end points, since other stops may lie in between. Thus, determining where an agent uses the bus may require knowing the full set of bus stop locations. We use this observation to prove the following.

\begin{theorem}\label{thm:time_hard_on_arbitrary_paths}
	\mbsp{} with objective function~$f_{\text{time}}$ or~$g_{\text{time}}$ is W[1]-hard on paths with arbitrary agent weights.
\end{theorem}
\begin{proof}
We use a reduction from \textsc{Independent Set} which is W[1]-hard for the parameter~$k$~\cite{DF13} to \mbsp{}.
An \textsc{Independent Set} instance consists of a graph~\(G = (V, E)\) and a budget~\(k\). The goal is to decide if there is an \emph{independent set}~\(X \subseteq V\) with~\(\abs{X} = k\) and~\(\{v,u\} \notin E\) for all~\(v,u \in X\).

Let $(G=(V,E),k)$ be an instance of \textsc{Independent Set}. Let $n = |V|$.
Let $\vec{V}=(v_1,v_2,\dots,v_{n})$ be any permutation of $V$.
Let $<$ be the strict linear order on $V$ defined by $\vec{V}$.
We construct an instance $(G',A,w_b,w_a,2k)$ of \mbsp{} as follows:

\begin{align*}
	V(G') ={}& \{X_{v}^{\circ} \mid v \in V, \circ \in \{-,+\}\} \\
	E(G') ={}& E_1 \cup E_2 \\
		E_1 ={}& \{\{X_{v_i}^{-}, X_{v_i}^{+}\} \mid i \in [n]\} \\
		E_2 ={}& \{\{X_{v_i}^{+}, X_{v_{i+1}}^{-}\} \mid i \in [n-1]\}  \end{align*}
		\begin{align*}
	A ={}& \{a^{(0)}_{u,v} = (X_{u}^{-}, X_{v}^{+}, w_{uv}) \mid \{u,v\} \in E, u < v\}\\ 
		\cup{}& \{a^{(1)}_{u,v} = (X_{v}^{-}, X_{v}^{+}, w_{uv}), a^{(2)}_{u,v} = (X_{u}^{-}, X_{u}^{+}, w_{uv}) \\
		& \mid u,v \in V, u < v, \{u,v\} \notin E\}\\
		w_b(e) ={}& 0 \text{ for } e \in E_1, \ \ \ w_b(e) = 3 \text{ for } e \in E_2 \\
		w_{uv}(e)={}& 1 \text{ for } e \in \{\{X_{u, -}, X_{u, +}\}, \{X_{v, -}, X_{v, +}\}\}, u < v \\
		w_{uv}(e)={}& 0 \text{ otherwise}
\end{align*}

In this construction $G'$ is a path between $X_{v_1}^{-}$ and $X_{v_{n}}^{+}$. Every vertex $v$ of the original graph $G$ is represented by an edge $\{X_{v_i}^{-}, X_{v_i}^{+}\}$ in $G'$. Every agent incurs a cost of at most two when travelling by foot since they travel over up to two edges of weight $1$. The only way to reduce that cost is to make them travel one of those edges by bus. The construction ensures that exactly $k$ such edges can be chosen to be travelled by bus. Since each of those edges represents a vertex in $G$, those chosen $k$ edges satisfy the objective function if and only if they represent an independent set of size $k$ in $G$.
\begin{claim}
	The independent set instance has a solution if and only if the objective function $g_{\text{time}}$ for instance constructed above is at most $\Sigma = (n-k) \cdot (n-1)$.
\end{claim}
\begin{claimproof}
    We say an edge $e \in E(G')$ is \emph{necessary} for an agent $a$ if $e$ is on the unique path between the terminals of $a$ and $w_a(e) \neq 0$. We say $a$ \emph{needs} the edge $e$. 
    Every necessary edge contributes to the cost of~$a$
    if and only if it is not traversed by bus.

	For every agent $a^{(i)}_{u,v}$ the total cost to travel its entire path without the bus is at most $2$. Thus, in all solutions minimizing $g_{\text{time}}$ no agent uses the bus to travel over an edge $e \in E_2$. Since every path of length $2$ or more in $G'$ contains an edge in $E_2$, in all solutions minimizing $g_{\text{time}}$, no agent uses the bus for more than one edge $e \in E_1$.

	Each edge~$\{X_{w}^{-}, X_{w}^{+}\}$ in~\(E_1\) is necessary for~\(n-1\) agents. One per pair \(w \neq v \in V\). Namely, it is necessary for an agent~\(a\) if~\(a \in \{a_{v,w}^{(0)}, a_{w,v}^{(0)}, a_{v,w}^{(1)}, a_{w,v}^{(2)} \mid v \in V(G)\}\). The total cost of necessary edges is $2 \cdot \binom{n}{2} = n \cdot (n-1)$.
    
    To use the bus to travel exactly one edge $e = \{X_{v}^{-}, X_{v}^{+}\} \in E_1$ both $X_{v}^{-}$ and $X_{v}^{+}$ have to be bus stops. Furthermore, every vertex is only part of one edge in $E_1$ meaning that at most $k$ edges in $E_1$ can be travelled by bus. Thus, there remain $n-k$ edges in $E_1$ who have to be travelled by foot by the $n-1$ agents who need them. The minimum cost of a solution is $(n-k) \cdot (n-1)$. As a result, the only way to achieve $g_{\text{time}}(B) \leq (n-k) \cdot (n-1)$ is to choose $k$ edges in $E_1$ such that every agent for whom those edges are necessary travels over that edge by bus. For a route \(B\), let the up to~\(k\) edges \(e \in E_1\) that can individually be travelled by bus be called \(E_B\). We define \(V_B = \{v \in V \mid \{X_{v}^{-}, X_{v}^{+}\} \in E_B\}\) as the set of corresponding vertices in $G$. \(B\) satisfies~$g_{\text{time}}(B) \leq (n-k) \cdot (n-1)$ if~\(|E_B| = k\) and every agent with 2 necessary edges(\(a_{u,v}^{(0)}\)) has at most one of them in \(E_B\). Let \(B\) be a bus route \(B\) such that~\(|E_B| = k\) and \(V_B\) is an independent set in~$G$. Then~since $V_B$ is an independent set, there is no agent \(a_{u,v}^{(0)}\) such that \(u,v \in V_B\). In contrast, if \(V_B\) is not an independent set in $G$, there exists an edge~\(\{u,v\} \in E, u,v \in V_B\) and thus both edges necessary for~\(a_{u,v}^{(0)}\) are in \(E_B\), meaning that $B$ does not satisfy~$g_{\text{time}}(B) \leq (n-k) \cdot (n-1)$.
\end{claimproof}

The proof extends to $f_{\text{time}}$ by adding bus stops at both ends of the path, ensuring that every agent can reach some stop without incurring cost.
\end{proof}

If we further restrict the problem to the unweighted case then we can solve \mbsp{} in polynomial time. The case for~$g_{\text{time}}$ is covered by Proposition~\ref{prop:g_time_poly_on_unweighted_graphs} while the case for~$f_{\text{time}}$ is covered by Theorem~\ref{thm:f_time_poly_on_unweighted_paths}.

\begin{theorem}\label{thm:f_time_poly_on_unweighted_paths}
	\mbsp{} with objective function~$f_{\text{time}}$ can be solved in polynomial time on paths in the unweighted case.
\end{theorem}

\begin{proof}
	To obtain an algorithm for $f_{\text{time}}$ on unweighted paths we can slightly modify the dynamic programming algorithm for~$f_{\text{energy}}$ from Theorem~\ref{thm:energy_poly_on_arbitrary_paths}. Since the bus costs are now counted per agent instead of once globally, we only need to change the main recurrence to the following:
	$\min_{i\in [j - 1]}
    \left( D[k' - 1, v_i] + f'_A(v_i, v_j) \right)$.

To calculate~$f'_A(v_i, v_j)$ we need to sum up the contribution of each individual agent. For this we consider four different types of agents~$(v_\ell, v_r, w)\in A$. Note that, since~$w_a = w_b$, an agent never needs to actually take the bus. They just need to walk to any bus stop and then to their end point.

\smallskip
\noindent
\emph{Case~$1$:~$i \leq \ell \leq r \leq j$.} In this case the agent starts and ends between~$v_i$ and~$v_j$. Here the agent needs to walk to one of the two bus stops~$v_i$ and~$v_j$. The contribution of such an agent is~$\min(\pi(v_\ell, v_i, w) + \pi(v_i, v_r, w), \pi(v_\ell, v_j, w) + \pi(v_j, v_r, w))$.

\smallskip
\noindent
\emph{Case~$2$:~$\ell < i \leq r \leq j$.} In this case only the end point is between~$v_i$ and~$v_j$. Here the agent can take the direct path from~$v_\ell$ to~$v_r$ since~$v_i$ is on the way. The contribution of such an agent is~$\pi(v_i, v_r, w)$.

\smallskip
\noindent
\emph{Case~$3$:~$i \leq \ell \leq j < r$.} In this case only the start point is between~$v_i$ and~$v_j$. Analogously to Case~$2$, the agent can take the direct path. The contribution of such an agent is~$\pi(v_r, v_j, w)$.

\smallskip
\noindent
\emph{Case~$4$:~$\ell < i \leq j < r$.} In this case neither start nor end point are between~$v_i$ and~$v_j$. Here the agent can also take the direct path. The contribution of such an agent is~$\pi(v_i, v_j, w)$.

The rest of the algorithm works the same as in Theorem~\ref{thm:energy_poly_on_arbitrary_paths}.
\end{proof}

\section{Experiments}

We evaluate the stop-selection algorithms on the NYC-M15 corridor~\cite{mta-developers}, restricting the candidate set to existing stops. This models a setting common in public transport planning: the underlying corridor already exists, and the task is to choose which stops should be visited by an additional express line. Origin-destination pairs for the agents are obtained from Citi Bike NYC Trip History~\cite{citibike-data}. We consider the scenario of adding an express line with~$k=10$ stops, where each agent has walking costs that are a factor~$w$ larger than the bus route costs and~$w$ is drawn uniformly at random from~$[1,2]$.

We compare three $k=10$ selections: one minimizing~$g_{\mathrm{energy}}$, one minimizing~$f_{\mathrm{energy}}$, and one using uniform spacing. Both optimized selections substantially outperform uniform spacing, reducing their respective objective costs by about~$16.5\%$ and~$16.8\%$. Figure~\ref{fig:m15-existing-stop-selection-comparison} shows that the selected stop sets depend on the objective. The two optimized solutions overlap in parts of the corridor, but they are not identical, illustrating that the two objectives capture different aspects of the stop-selection problem. In particular, $g_{\mathrm{energy}}$, which allows agents to walk directly, gives a more compact stop placement because it does not produce artificially high costs for areas without bus stops. Since the preferred stop set changes with the objective, both objectives are important to consider when evaluating stop-selection quality.

\definecolor{ColorOpen}{RGB}{255,255,255}
\definecolor{ColorW2h}{RGB}{255,100,100}   
\definecolor{ColorW1h}{RGB}{255,200,100}   
\definecolor{ColorNPh}{RGB}{255,255,100}   
\definecolor{ColorP}{RGB}{150,200,50}      
\definecolor{ColorUnclear}{RGB}{180,180,180}

\begin{figure}[t]
	\centering
	
	\includegraphics[width=0.47\textwidth]{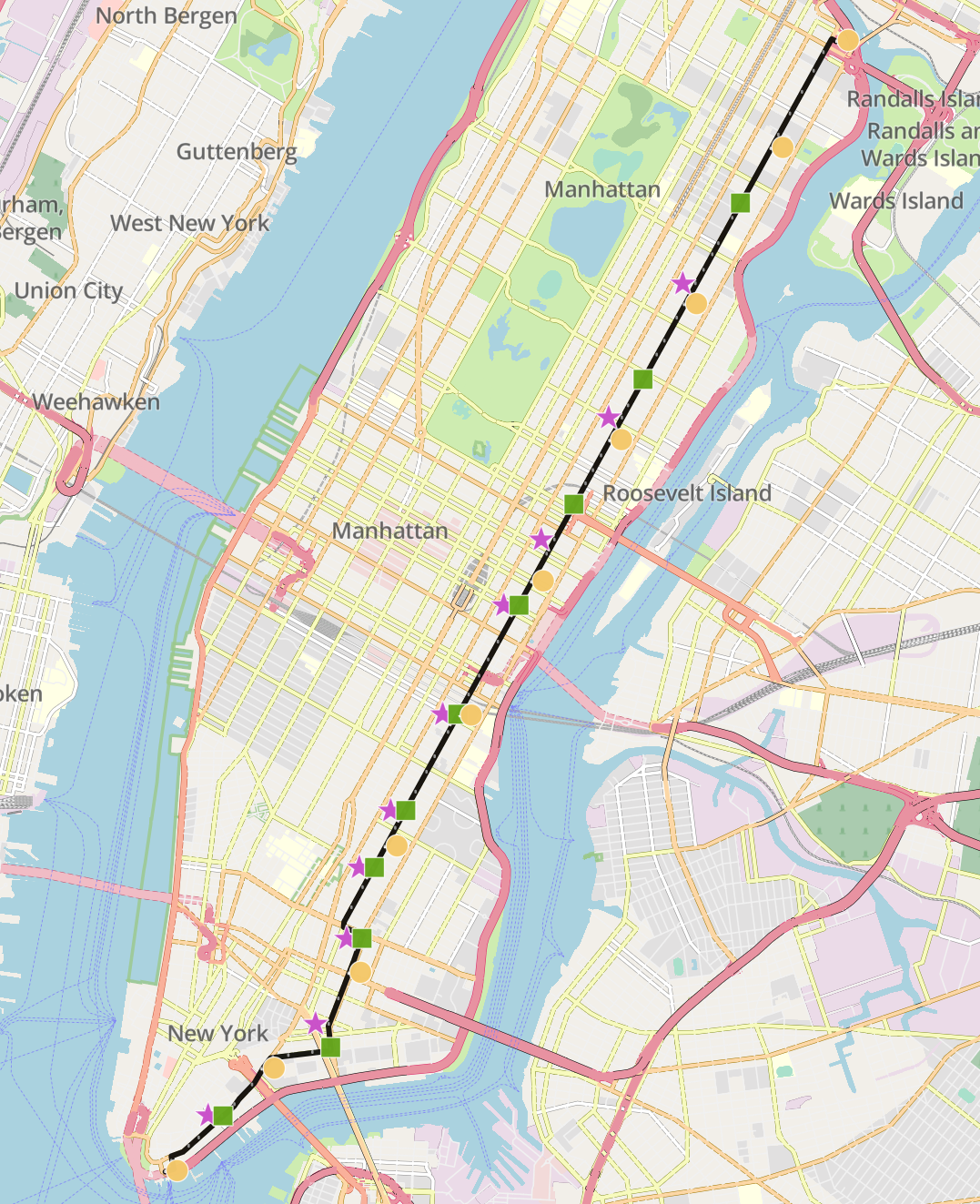}
	
	\vspace{0.8em}

	\definecolor{Purple}{RGB}{201,80,201}
	
	\resizebox{0.48\textwidth}{!}{%
		\centering
		\begin{tikzpicture}[
			x=1.15cm,
			y=1cm,
			gopt/.style={star, star point ratio=2.8, fill=Purple, draw=Purple, inner sep=0.9pt},
			fopt/.style={rectangle, fill=ColorP, draw=ColorP, inner sep=2.2pt},
			uniform/.style={circle,
				fill=ColorW1h, draw=ColorW1h, inner sep=2.0pt},
			legendlabel/.style={anchor=west},
			]
			
			\begin{scope}[xscale=-1, xshift=-0.458\textwidth] 
				\draw[black, line width=0.75pt] (0,0) -- (7.2,0);
				
				\foreach \xx in {0.000,0.108,0.217,0.342,0.470,0.596,0.677,1.006,1.088,1.239,1.304,1.396,1.477,1.602,1.680,1.816,1.898,2.050,2.201,2.271,2.407,2.495,2.583,2.661,2.790,2.994,3.078,3.245,3.386,3.599,3.705,3.801,3.912,4.034,4.134,4.278,4.369,4.487,4.602,4.702,4.814,4.939,5.053,5.173,5.315,5.422,5.610,5.751,5.892,6.021,6.128,6.254,6.335,6.420,6.517,6.655,6.743,6.908,6.975,7.107,7.200} {
					\draw[gray!55, line width=0.3pt] (\xx,-0.16) -- (\xx,0.16);
				}
				
				\foreach \xx in {1.477,2.271,2.994,3.386,4.034,4.602,4.939,5.422,5.892,6.743} {
					\node[gopt] at (\xx,0.13) {};
				}
				
				\foreach \xx in {1.006,2.050,2.790,3.386,4.034,4.602,4.939,5.422,6.021,6.743} {
					\node[fopt] at (\xx,0.00) {};
				}
				
				\foreach \xx in {0.000,0.677,1.602,2.407,3.245,4.034,4.814,5.610,6.420,7.200} {
					\node[uniform] at (\xx,-0.13) {};
				}
			\end{scope}
			
			\node[gopt] at (-0.15,-0.75) {};
			\node[font=\small,legendlabel] at (-0.02,-0.75) {$g_{\mathrm{energy}}$-optimal}; 
			
			\node[fopt] at (2.05,-0.75) {};
			\node[font=\small,legendlabel] at (2.18,-0.75) {$f_{\mathrm{energy}}$-optimal}; 
			
			\node[uniform] at (4.25,-0.75) {};
			\node[font=\small,legendlabel] at (4.38,-0.75) {uniform};
			
			\draw[gray!55, line width=0.3pt] (5.65,-0.9) -- (5.65,-0.6);
			\node[font=\small, anchor=west] at (5.75,-0.75) {existing stops};
			
		\end{tikzpicture}%
	}
	\caption{Comparison of three $k=10$ stop selections on the NYC-M15 corridor, restricted to existing stops.}
	\label{fig:m15-existing-stop-selection-comparison}
\end{figure}
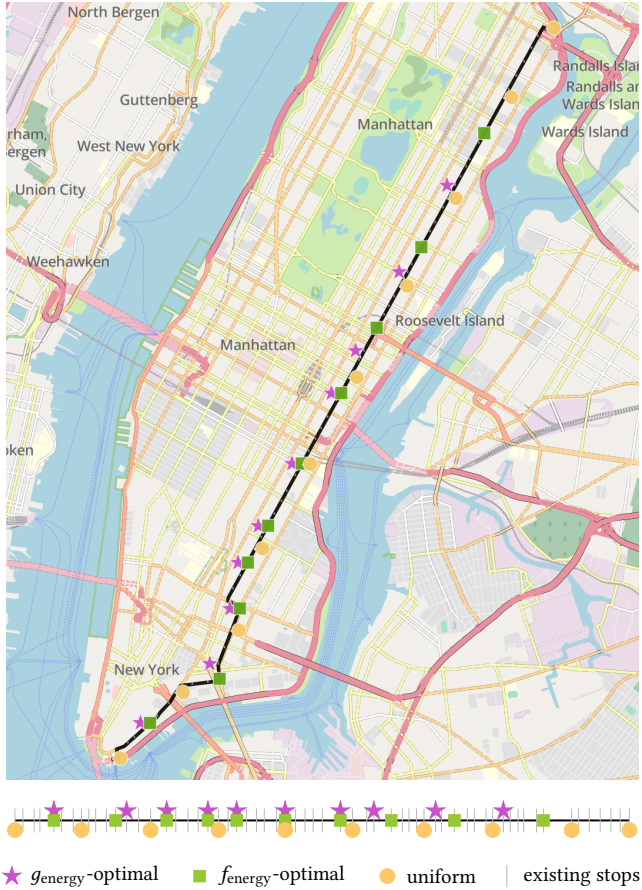


\section{Conclusion}

There are several possibilities to expand our results. Clearly, the complexity of the four remaining path cases still needs to be resolved.  Another direction is to generalize the positive results obtained for paths to larger classes of graphs. However, for all objective functions except $f_\text{energy}$ we have hardness already for stars, which excludes positive results for graphs with pathwidth~2, the most obvious next step beyond paths. Thus, even stronger restrictions such as bounded maximum degree seem necessary to extend the positive results. 

Another and possibly more fruitful direction could be to consider further types of agent models, for example the case where edge weights for the agents are never smaller than bus-edge weights. One could also study an intermediate model between arbitrary weights and consistent weights, where there is only a bounded number of distinct agent weight functions. Finally, a parameter that is typically large in the instances constructed by our reductions is the number of agents. Thus, it remains unclear whether the problem is still hard when this parameter is small. From a theoretical perspective, it would be interesting to determine whether FPT algorithms parameterized by the number of agents exist.

\section*{Acknowledgements}

Jurek Rostalsky and Luca Pascal Staus
were supported by the Carl Zeiss Foundation, Germany,
within the project “Interactive Inference”.
Eva Deltl acknowledges her support by the Deutsche Forschungsgemeinschaft (German Research Foundation, DFG), project COMSOC-MPMS (grant agreement No.\ 465371386) and project MaMu (grant agreement No.\ 392018064).
\bibliographystyle{named}
\bibliography{refs}

@inproceedings{DBLP:conf/gis/ProisslK24,
  author       = {Claudius Proissl and
                  Daniel Koch},
  title        = {Revisiting the Bus Stop Problem in Road Networks},
  booktitle    = {Proceedings of the 32nd {ACM} International Conference on Advances
                  in Geographic Information Systems, {SIGSPATIAL}~'24, Atlanta, GA,
                  USA},
  pages        = {80--90},
  publisher    = {{ACM}},
  year         = 2024
}

@article{Tamir96,
  author       = {Arie Tamir},
  title        = {An $O(p n^2)$ algorithm for the p-median and related problems
                  on tree graphs},
  journal      = {Oper. Res. Lett.},
  volume       = {19},
  number       = {2},
  pages        = {59--64},
  year         = {1996},
}

@article{EP99,
author = {M. G. Everett and S. P. Borgatti},
title = {The centrality of groups and classes},
journal = {The Journal of Mathematical Sociology},
volume = {23},
number = {3},
pages = {181--201},
year = {1999},
publisher = {Routledge},
doi = {10.1080/0022250X.1999.9990219},
URL = { 
        https://doi.org/10.1080/0022250X.1999.9990219
},
eprint = { 
        https://doi.org/10.1080/0022250X.1999.9990219
}
}

@inproceedings{SKMS23,
  author       = {Luca Pascal Staus and
                  Christian Komusiewicz and
                  Nils Morawietz and
                  Frank Sommer},
  editor       = {Jonathan W. Berry and
                  David B. Shmoys and
                  Lenore Cowen and
                  Uwe Naumann},
  title        = {Exact Algorithms for Group Closeness Centrality},
  booktitle    = {{SIAM} Conference on Applied and Computational Discrete Algorithms,
                  {ACDA} 2023, Seattle, WA, USA, May 31 - June 2, 2023},
  pages        = {1--12},
  publisher    = {{SIAM}},
  year         = {2023},
  url          = {https://doi.org/10.1137/1.9781611977714.1},
  doi          = {10.1137/1.9781611977714.1},
  bibsource    = {dblp computer science bibliography, https://dblp.org}
}

@inproceedings{ACL+20,
  author       = {Haris Aziz and
                  Hau Chan and
                  Barton Lee and
                  Bo Li and
                  Toby Walsh},
  title        = {Facility Location Problem with Capacity Constraints: Algorithmic and
                  Mechanism Design Perspectives},
  booktitle    = {The Thirty-Fourth {AAAI} Conference on Artificial Intelligence, {AAAI}
                  2020, The Thirty-Second Innovative Applications of Artificial Intelligence
                  Conference, {IAAI} 2020, The Tenth {AAAI} Symposium on Educational
                  Advances in Artificial Intelligence, {EAAI} 2020},
  pages        = {1806--1813},
  publisher    = {{AAAI} Press},
  year         = {2020},
  url          = {https://doi.org/10.1609/aaai.v34i02.5547},
  doi          = {10.1609/AAAI.V34I02.5547},
  bibsource    = {dblp computer science bibliography, https://dblp.org}
}

@inproceedings{originalBus, author = {Radi Muhammad Reza and Mohammed Eunus Ali and Muhammad Aamir Cheema},
 title = {The Optimal Route and Stops for a Group of Users in a Road Network}, year = {2017}, isbn = {9781450354905}, publisher = {Association for Computing Machinery}, address = {New York, NY, USA}, url = {https://doi.org/10.1145/3139958.3140061}, doi = {10.1145/3139958.3140061}, booktitle = {Proceedings of the 25th ACM SIGSPATIAL International Conference on Advances in Geographic Information Systems}, articleno = {4}, numpages = {10}, location = {Redondo Beach, CA, USA}, series = {SIGSPATIAL'17} }

@book{DF13,
  author       = {Rodney G. Downey and
                  Michael R. Fellows},
  title        = {Fundamentals of Parameterized Complexity},
  series       = {Texts in Computer Science},
  publisher    = {Springer},
  year         = {2013},
  url          = {https://doi.org/10.1007/978-1-4471-5559-1},
  doi          = {10.1007/978-1-4471-5559-1},
  isbn         = {978-1-4471-5558-4},
  bibsource    = {dblp computer science bibliography, https://dblp.org}
}

@book{CFK+15,
  author       = {Marek Cygan and
                  Fedor V. Fomin and
                  Lukasz Kowalik and
                  Daniel Lokshtanov and
                  D{\'{a}}niel Marx and
                  Marcin Pilipczuk and
                  Michal Pilipczuk and
                  Saket Saurabh},
  title        = {Parameterized Algorithms},
  publisher    = {Springer},
  year         = {2015},
}

@article{KOC2020104987,
	title = {A review of vehicle routing with simultaneous pickup and delivery},
	journal = {Computers \& Operations Research},
	volume = {122},
	pages = {104987},
	year = {2020},
	issn = {0305-0548},
	doi = {https://doi.org/10.1016/j.cor.2020.104987},
	url = {https://www.sciencedirect.com/science/article/pii/S0305054820301040},
author = {{\c{C}}a{\u{g}}r{\i} Ko{\c{c}} and Gilbert Laporte and {\.I}lknur T{\"u}kenmez}
}

@Inbook{Ljubic20,
author="Ljubi{\'{c}}, Ivana",
title="Connected Facility Location Problems",
bookTitle="Encyclopedia of Optimization",
year="2020",
publisher="Springer International Publishing",
address="Cham",
pages="1--11",
isbn="978-3-030-54621-2",
doi="10.1007/978-3-030-54621-2_870-1",
url="https://doi.org/10.1007/978-3-030-54621-2_870-1"
}

@techreport{CNW83,
  title={The uncapacitated facility location problem},
  author={Cornu{\'e}jols, G{\'e}rard and Nemhauser, George and Wolsey, Laurence},
  year={1983},
  institution={Cornell University Operations Research and Industrial Engineering}
}

@article{FELLOWS20111118,
	title = {Facility location problems: A parameterized view},
	journal = {Discrete Applied Mathematics},
	volume = {159},
	number = {11},
	pages = {1118-1130},
	year = {2011},
	issn = {0166-218X},
	doi = {https://doi.org/10.1016/j.dam.2011.03.021},
	url = {https://www.sciencedirect.com/science/article/pii/S0166218X11001156},
	author = {Michael R. Fellows and Henning Fernau}
}

@article{AG20,
  author       = {Jonnatan Fernando Avil{\'{e}}s{-}Gonz{\'{a}}lez and
                  Jaime Mora{-}Vargas and
                  Neale R. Smith and
                  Miguel Gast{\'{o}}n Cedillo{-}Campos},
  title        = {Artificial intelligence and {DOE:} an application to school bus routing
                  problems},
  journal      = {Wireless Networks},
  volume       = {26},
  number       = {7},
  pages        = {4975--4983},
  year         = {2020},
}

@article{SKS+13,
  author       = {Patrick Schittekat and
                  Joris Kinable and
                  Kenneth S{\"{o}}rensen and
                  Marc Sevaux and
                  Frits C. R. Spieksma and
                  Johan Springael},
  title        = {A metaheuristic for the school bus routing problem with bus stop selection},
  journal      = {European Journal of Operational Research},
  volume       = {229},
  number       = {2},
  pages        = {518--528},
  year         = {2013},
  url          = {https://doi.org/10.1016/j.ejor.2013.02.025},
  doi          = {10.1016/J.EJOR.2013.02.025},
  bibsource    = {dblp computer science bibliography, https://dblp.org}
}

@article{PK10,
  author       = {Junhyuk Park and
                  Byung{-}In Kim},
  title        = {The school bus routing problem: {A} review},
  journal      = {European Journal of Operational Research},
  volume       = {202},
  number       = {2},
  pages        = {311--319},
  year         = {2010},
  url          = {https://doi.org/10.1016/j.ejor.2009.05.017},
  doi          = {10.1016/J.EJOR.2009.05.017},
  bibsource    = {dblp computer science bibliography, https://dblp.org}
}

@inproceedings{DBLP:conf/coco/Karp72,
  author       = {Richard M. Karp},
  title        = {Reducibility Among Combinatorial Problems},
  booktitle    = {Complexity of Computer Computations},
  series       = {The {IBM} Research Symposia Series},
  pages        = {85--103},
  publisher    = {Plenum Press, New York},
  year         = {1972}
}

@article{XiaoEtal2022, author = {Xiao, Mingyu and Zhang, Jianan and Lin, Weibo}, title = {Parameterized algorithms and complexity for the traveling purchaser problem and its variants}, year = {2022}, issue_date = {Nov 2022}, publisher = {Springer-Verlag}, address = {Berlin, Heidelberg}, volume = {44}, number = {4}, issn = {1382-6905}, url = {https://doi.org/10.1007/s10878-020-00608-x}, doi = {10.1007/s10878-020-00608-x}, journal = {Journal of Combinatorial Optimization}, month = nov, pages = {2269–2285}, numpages = {17} }

@article{MANERBA2017,
	title = {The Traveling Purchaser Problem and its variants},
	journal = {European Journal of Operational Research},
	volume = {259},
	number = {1},
	pages = {1-18},
	year = {2017},
	issn = {0377-2217},
	doi = {https://doi.org/10.1016/j.ejor.2016.12.017},
	url = {https://www.sciencedirect.com/science/article/pii/S0377221716310517},
	author = {Daniele Manerba and Renata Mansini and Jorge Riera-Ledesma}
}

@article{HK79,
  title={An algorithmic approach to network location problems II: The $p$-medians},
  author={Hakimi, S and Kariv, O},
  journal={SIAM Journal on Applied Mathematics},
  volume={37},
  number={3},
  pages={539--560},
  year={1979}
}

@article{ringstar1,
	author = {Labbé, Martine and Laporte, Gilbert and Martín, Inmaculada Rodríguez and González, Juan José Salazar},
	title = {The Ring Star Problem: Polyhedral analysis and exact algorithm},
	journal = {Networks},
	volume = {43},
	number = {3},
	pages = {177-189},
	doi = {https://doi.org/10.1002/net.10114},
	url = {https://onlinelibrary.wiley.com/doi/abs/10.1002/net.10114},
	eprint = {https://onlinelibrary.wiley.com/doi/pdf/10.1002/net.10114},
	year = {2004}
}

@article{MORADI2024110730,
	title = {Set Covering Routing Problems: A review and classification scheme},
	journal = {Computers \& Industrial Engineering},
	volume = {198},
	pages = {110730},
	year = {2024},
	issn = {0360-8352},
	doi = {https://doi.org/10.1016/j.cie.2024.110730},
	url = {https://www.sciencedirect.com/science/article/pii/S0360835224008520},
	author = {Nima Moradi and Fereshteh Mafakheri and Chun Wang},
}

@article{GuihaireHao08,
	author  = {Guihaire, Val{\'e}rie and Hao, Jin-Kao},
	title   = {Transit network design and scheduling: A global review},
	journal = {Transportation Research Part A: Policy and Practice},
	volume  = {42},
	number  = {10},
	pages   = {1251--1273},
	year    = {2008},
	doi     = {10.1016/j.tra.2008.03.011}
}

@article{Owais26,
	author  = {Owais, Mahmoud},
	title   = {Transit network design problem: a half century of methodological research},
	journal = {Innovative Infrastructure Solutions},
	volume  = {11},
	number  = {3},
	year    = {2026},
	doi     = {10.1007/s41062-025-02356-5}
}

@article{DerribleKennedy11,
	author  = {Derrible, Sybil and Kennedy, Christopher},
	title   = {Applications of Graph Theory and Network Science to Transit Network Design},
	journal = {Transport Reviews},
	volume  = {31},
	number  = {4},
	pages   = {495--519},
	year    = {2011},
	doi     = {10.1080/01441647.2010.543709}
}

@article{FurthRahbee00,
	author  = {Furth, Peter G. and Rahbee, Adam B.},
	title   = {Optimal Bus Stop Spacing Through Dynamic Programming and Geographic Modeling},
	journal = {Transportation Research Record},
	volume  = {1731},
	number  = {1},
	pages   = {15--22},
	year    = {2000},
	doi     = {10.3141/1731-03}
}

@online{citibike-data,
	title   = {Citi Bike System Data},
	author  = {{Citi Bike}},
	url     = {https://citibikenyc.com/system-data},
	urldate = {2026-06-05},
	note={\url{https://citibikenyc.com/system-data}}
}

@online{mta-developers,
	title   = {MTA Developer Resources},
	author  = {{Metropolitan Transportation Authority}},
	url     = {https://www.mta.info/developers},
	urldate = {2026-06-05},
	note={\url{https://www.mta.info/developers}}
}

\end{document}